\documentclass[11pt]{article}
\usepackage{amsmath,amssymb,amsfonts} 
\usepackage{epsfig} \usepackage{latexsym,nicefrac,bbm}
\usepackage{xspace}
\usepackage{color,fancybox,graphicx,subfigure,fullpage}
\usepackage[top=1in, bottom=1in, left=1in, right=1in]{geometry}
\usepackage{tabularx} \usepackage{hyperref} 
\usepackage{pdfsync}
\usepackage[boxruled]{algorithm2e}
\usepackage{multicol}
\usepackage{enumitem}
\usepackage{multirow}

\renewcommand{\epsilon}{\varepsilon}

\newcommand{\poly}{\mathrm{poly}}

\newcommand\N{\mathbb N}
\newcommand\R{\mathbb R}
\newcommand\C{\mathbb C}

\newtheorem{theorem}{Theorem}[section]

\newtheorem{definition}{Definition}[section]
\newtheorem{lemma}[theorem]{Lemma}
\newtheorem{remark}[theorem]{Remark}
\newtheorem{proposition}[theorem]{Proposition}
\newtheorem{corollary}[theorem]{Corollary}

\newenvironment{proof}{\begin{trivlist} \item {\bf Proof:~~}}
   {\qed\end{trivlist}}

\def\FullBox{\hbox{\vrule width 6pt height 6pt depth 0pt}}

\def\qed{\ifmmode\qquad\FullBox\else{\unskip\nobreak\hfil
\penalty50\hskip1em\null\nobreak\hfil\FullBox
\parfillskip=0pt\finalhyphendemerits=0\endgraf}\fi}

\renewcommand{\C}{\mathbb{C}}
\renewcommand{\R}{\mathbb{R}}
\renewcommand{\N}{\mathbb{N}}

\DeclareMathOperator{\Tr}{Tr}

\DeclareMathOperator{\supp}{supp}
\DeclareMathOperator{\diag}{diag}

\DeclareMathOperator{\hull}{hull}

\DeclareMathOperator{\primal}{Prim}
\DeclareMathOperator{\dual}{Dual}

\usepackage[T1]{fontenc}
\usepackage{lmodern}

\title{\bf On the Computability of Continuous Maximum Entropy Distributions with Applications}
\author{Jonathan Leake \\ KTH \and Nisheeth K. Vishnoi \\ Yale University}

\begin{document}

\maketitle

\sloppy

\begin{abstract}

    We initiate a study of the following  problem:
    Given a continuous domain $\Omega$ along with its convex hull  $\mathcal{K}$, a point $A \in \mathcal{K}$ and a prior measure $\mu$ on $\Omega$, find the probability density over $\Omega$ 
    whose marginal is $A$ and that minimizes the KL-divergence to  $\mu$.
    This framework gives rise to several extremal distributions that  arise
    in mathematics, quantum mechanics, statistics, and theoretical computer science. 
    Our  technical contributions  include a polynomial bound on the norm of the optimizer of the dual problem that holds in a very general setting and relies on a  ``balance'' property of the measure $\mu$ on $\Omega$, and exact algorithms for evaluating the dual and its gradient for several interesting settings of $\Omega$ and $\mu$.
    Together, along with the ellipsoid method, these results imply polynomial-time algorithms to compute such KL-divergence minimizing distributions in several cases.
    Applications of our results include:
    1) an optimization characterization of  the Goemans-Williamson measure \cite{GoemansW1995} that is used to round a positive semidefinite matrix to a vector,
    2) the computability of the entropic barrier for polytopes studied  by  \cite{BubeckEldan}, and
    3) a polynomial-time algorithm to compute the barycentric quantum entropy of a density matrix that was proposed as an alternative to von Neumann entropy in the 1970s  \cite{Band1976,Park1977,SlaterEntropy1}:
    this corresponds to the case when  $\Omega$ is the set of rank one projections matrices and $\mu$ corresponds to the  Haar measure on the  unit sphere.
    Our techniques generalize to the setting of Hermitian rank $k$ projections using the Harish-Chandra-Itzykson-Zuber formula \cite{HarishChandra1957,IZ1980}, and are applicable even beyond,  to adjoint orbits of compact Lie groups.

\end{abstract}

\newpage

\tableofcontents
\newpage

\section{Introduction}

\paragraph{Entropy maximizing distributions.}
Let $\Omega$ be a subset of $\R^d$ and
let $\mathcal{K}=\mathrm{hull}(\Omega)$ denote the convex hull of $\Omega$.
Suppose one is given an $A \in \mathcal{K}$. 
A natural question arises: {\em Is there a canonical way to choose a probability measure supported on $\Omega$ that can be used to express $A$ as a convex combination of points on $\Omega$?} 
When $\Omega$ is a discrete and finite set, this problem has been extensively studied and a canonical probability distribution was proposed by Jaynes \cite{Jaynes1,Jaynes2}: among all probability distributions that can be used to express $A$ as a convex combination of points in $\Omega$,  pick the one that maximizes the Shannon entropy.
These distributions are referred to as maximum entropy (max-entropy) distributions and arise in machine learning, statistics, mathematics, and theoretical computer science (TCS).
In TCS, these distributions have found many uses due to duality, connections to polynomials, and  algorithms to compute them \cite{GS02,SinghV14,asadpour2017atsp,garg2015,CMTV17,ALOW17}; see \cite{StraszakV19}.

In this paper we initiate a study of the computability when $\Omega$ is  a continuous (and often nonconvex) manifold.
Examples of interest include
\[
    \mathcal{V}_1:=\{ vv^\top: v \in \mathbb{R}^n \},
\]
\[
    \mathcal{P}_1:= \{ vv^*: v \in \C^n, \|v\|_2=1\},
\]
the set of rank $k$ Hermitian projection matrices
\[
    \mathcal{P}_k:= \{ Y: Y \in \C^{n \times n}, \Tr(Y)=k, Y=Y^*, Y^2=Y\}
\]
(related to the Grassmanian), or a convex body (in which case $\mathcal{K}=\Omega).$

Unlike the discrete setting, in the continuous setting the notion of finding a max-entropy distribution is not well-defined since a canonical notion of entropy does not necessarily exist. 
We  instead consider relative entropy, Kullback-Leibler (KL) divergence with respect to a prior measure $\mu$ on $\Omega$ that corresponds to the density function $f(X)\equiv 1$ for all $X \in \Omega.$
For all of the manifolds mentioned above, there is a canonical measure that has this property and is called the uniform measure; see Section \ref{sec:preliminaries}. 
This leads us to the following infinite dimensional convex optimization problem which gives a canonical way to write $A$ as a convex combination of points in $\Omega$: Find a measure $\nu$ on $\Omega$ that is continuous with respect to $\mu$ and, subject to the constraint that the expected point in $\mathcal{K}$ with respect to $\nu$ is $A$,  $\nu$ minimizes the KL divergence to $\mu.$ 
Note that, by choice, $\nu$ is as close to the distribution $\mu$ as possible; hence we call it a maximum entropy distribution.

The class of extremal entropy maximizing distributions that arise in this manner have several properties that have led  to their appearance, implicitly or explicitly,  in several different areas: %

\begin{itemize}[leftmargin=2mm]
\setlength\itemsep{0em}
\item the work of Klartag (inspired by a work of Gromov) on  the isotropic constant \cite{Klartag2006,Gromov1990},
\item the work of  Khatri and Mardia on the Matrix Bingham distribution  in statistics with applications to various scientific and engineering problems \cite{bingham1974,Khatri77,Hoff2009}, 
\item as shown  here, the work of Goemans and Williamson on rounding semidefinite programs \cite{GoemansW1995},
\item  the works of G\"uler, Bubeck and Eldan on barrier functions for interior point methods \cite{Guler1997,Guler1998,BubeckEldan}, 
\item the works of Band, Park, and Slater that defined the barycentric quantum entropy and proposed it as an alternative to the von Neumann entropy in the 1970s \cite{Band1976,Park1977,SlaterEntropy1}.

\end{itemize}

\noindent
\paragraph{Computability of entropy maximizing distributions.} One of the reasons why the entropy maximizing problem defined earlier is interesting (and unifies the above problems)  is {\em duality}: the dual optimization problem roughly has the form: 
$$ \inf_{Y} \ \langle Y,A\rangle + \log \int_{X \in \Omega} e^{-\langle Y, X \rangle } d\mu(X),$$
where $\langle \cdot, \cdot \rangle$ is an inner product and $\mu$ is the given  measure.
If strong duality holds, it can be shown that the optimal distribution $\nu^\star$ to the entropy maximizing problem can be described by the optimizer $Y^\star$ to the dual above:
$ \textstyle \nu^\star (X) \propto e^{-\langle Y^\star,X \rangle}$  for $X \in \Omega.$
As for computability of $\nu^\star$, $Y^\star$ lives in a small, convex, and finite dimensional (same dimension as $\mathcal{K}$) domain.
Hence in principle, one could hope to represent $\nu^\star$ efficiently.
However, bounding the running time of a   optimization method to find $Y^\star$ reduces to 1) a bounding some norm of $Y^\star$ and, 2) coming up with efficient algorithms to compute  $\int_{X \in \Omega} e^{-\langle Y, X \rangle } d\mu(X)$ for matrices $Y$ with that norm.
These are the main problems studied in this paper.

\subsection{Our contributions}

The main contributions of this paper are to initiate a formal study of the computability of entropy maximizing distributions on continuous domains, to present an ellipsoid method-based framework to compute them, to derive polynomial time algorithms for computing  maximum entropy distributions for specific manifolds mentioned earlier, and to present implications to  some of the applications  listed above.

\paragraph{The continuous maximum entropy framework and duality.} Our general framework is presented Section \ref{sec:maxentropy}.
The focus is on the setting when the manifold $\Omega$ and the base measure $\mu$  is fixed to either the set of all rank one matrices over reals ($\mathcal{V}_1$) with the measure induced by Lebesgue measure on $\R^n$, or the set of all rank $k$ projections over complexes ($\mathcal{P}_k$) for $k \geq 1$ with the appropriate Haar measure.
The input consists of an element $A$ (which is a matrix in the cases of interest) and the goal is to compute a representation for $\nu^\star$ that is the KL-divergence minimizing distribution to $\mu$ with marginal $A$.
We start by writing down the dual of this optimization problem (Section \ref{sect:dualform}) and showing  that strong duality holds under Slater's condition  -- that there is a density function that is strictly positive (and bounded) on $\Omega$ and has marginal $A$ (Section \ref{sect:sd_proof}).
This is implied by the condition that  $A$ is in the relative interior of the convex hull $\mathcal{K}$ of $\Omega$, which we then show is true quite generally in Sections \ref{sec:general_sd} and \ref{sec:interior}.
Strong duality then implies that the optimal measure $\nu^\star$ is determined by the optimal dual solution $Y^\star$ as $\nu^\star(X) \propto e^{-\langle Y^\star,X \rangle}$; see Theorem \ref{thm:strongduality}.

\paragraph{Norm of the optimal dual solution.}
However, to solve the dual convex program one needs, at the bare minimum, that the norm of $Y^\star$ is reasonably bounded.
It is not difficult to see that as  $A$ tends to the boundary of $\mathcal{K}$, the optimal measure is concentrated on a face of $\mathcal{K}$ implying that the norm of $Y^\star$ must tend to infinity.
Thus, one needs some assumption on the ``interiority'' of $A$ to ensure polynomial time computability. 
The situation is exacerbated by the fact that the  $Y^\star$ appears in the exponent and, hence, to have any hope of computability of the entropy maximizing distribution, the bound  on $Y^\star$ should be {\em polynomial} in the bit complexity of $A$.
Unlike the case when $\Omega$ is discrete (studied in \cite{SinghV14}), the fact that the base measure $\mu$ is continuous makes it harder.
Our main contribution towards the problem of bounding the norm of $Y^\star$ involves identifying a certain ``balance'' property of the measure $\mu$ on the manifold $\Omega$ (Definition \ref{def:balanced}) and showing that, roughly,   $\|Y^\star\| \le \mathrm{poly}(d,1/\eta)$ where $\eta$  is the distance of $A$ from the boundary of $\mathcal{K}$; see Theorem \ref{thm:boundingbox}.
We show that this balance property holds for a wide class of manifolds and obtain as corollaries  a bound of $\mathrm{poly}(n,1/\eta)$ for both $\Omega = \mathcal{P}_k$ (Corollary \ref{cor:boundingbox_rank_k}) and when $\Omega$ is an $n$-dimensional convex body (Corollary \ref{cor:boundingbox_convex}).
This bounding box result is quite general and expected to find further applications.  

\paragraph{Computing the integral in the dual for matrix manifolds.}
A bound on the norm of $Y^\star$ allows us to show that we can use the  ellipsoid method to solve the dual convex program, provided the measure $\mu$ is balanced on $\Omega$, and we can evaluate the dual and its gradient at a specified $Y$ of norm up to that of  $Y^\star$.
The tasks of evaluating the dual and its gradient essentially reduce to the computation of the integral
$\int_{X \in \Omega} e^{-\langle Y, X \rangle } d\mu(X).$
In the case when $\Omega = \mathcal{V}_1$ with $\mu$ being the measure induced by the Lebesgue measure, we observe that the dual optimization problem is finite only when $Y \succ 0$, and thus we need to evaluate the integral only for such a $Y$.  
The integral above then turns out to have a simple formula: roughly, $\log \det Y$ (Proposition \ref{prop:lebesgue_formula}). 

In other interesting cases, computing such an integral turns out to be a  nontrivial task.
In the case when $\Omega = \mathcal{P}_1$ and $\mu_1$ is the uniform measure induced by the Haar measure on the  complex unit sphere, we first note that the entropy maximizing measure cannot be obtained by solving the problem first for $\mathcal{V}_1$ and then ``projecting'' it on the sphere; see Section \ref{sec:GWsphere}.
Then, we note that the integral does not reduce to a product of $n$ integrals as in the Lebesgue case, and there is no easy way around this.
We need an algorithm to integrate the density $e^{-v^* Y v}$ over the complex unit sphere where the only thing we know about $Y$ is that it is Hermitian.
Neither the density is log-concave, nor the support (unit sphere) is convex.
Our main contribution here is to give an {\em exact} algorithm to compute this integral whose running time depends {\em single exponentially} on the bit complexity of the input $Y$ to it (Theorem \ref{thm:counting}). 
As remarked earlier, because $Y$ is being exponentiated, this is the best one can hope for and also turns out to be sufficient to obtain polynomial time algorithms for computing maximum entropy distributions on $\mathcal{P}_1.$

Interestingly, the algorithm to compute this integral and its proof relies on an  connection between the manifold $\mathcal{P}_1$ and
the probability simplex in $n$ dimensions.
Specifically, one can naturally push forward the entropy maximizing measure from $\mathcal{P}_1$ to a log-linear measure on the corresponding simplex.
There are then algorithms to sample from such a density function on the simplex to estimate such an integral; however, to obtain an $1+\delta$ approximation to it, the running time of these methods depends polynomially on $1/\delta$ instead  $\log 1/\delta$. 
We give an exact algorithm to compute this integral. 
Our method relies on Laplace transforms, is elementary, and
a significant effort is needed to deal with the case when  $Y$ has repeated eigenvalues.
Importantly, this viewpoint also leads us to an exact algorithm for computing such an integral for $\mathcal{P}_k$ for $k>1$ using the  Harish-Chandra-Itzykson-Zuber formula \cite{HarishChandra1957,IZ1980,DuistermaatH1982,Vergne1996}; see Theorem \ref{thm:HCIZ}.

\paragraph{Efficient algorithm via the ellipsoid method.}
Our general ellipsoid method-based algorithm  requires  1) a  full dimensional embedding of hull($\Omega$) in a $d$-dimensional real Hilbert space, 2)  $\mu$ is a balanced measure on $\Omega$, 3) $\Omega$ is contained in a ball of radius $r$, 4) the point $A$ is in the $\eta$-interior of hull($\Omega$) and, 5) that we have an exact counting/integrating oracle.
It runs in time polynomial in $d,1/\eta,\log r$ and $\log 1/\epsilon,$ to solve the dual problem to an additive $\epsilon$; see Theorem \ref{thm:main_algorithm}.
Our bound on the norm of  $Y^\star$ and exact algorithms to compute the dual objective/gradient for  the case of $\mathcal{P}_k$ imply a polynomial time algorithm to compute the entropy maximizing measure in this case when $A$ is in the polynomial interior of hull($\mathcal{P}_k$); see Corollary \ref{cor:maxentropy_algo}.

\subsection{Applications}

\paragraph{SDP rounding.}
One approach to semi-definite programming (SDP) based approximation algorithms,  starting with the work of Goemans-Williamson \cite{GoemansW1995} for the maximum cut problem, is SDP rounding.
Here, typically, $A$ is a positive semi-definite (PSD) matrix, that is computed using a SDP relaxation to some non-convex problem, and one of the goals is to {\em round} $A$ to a vector.
This involves choosing a distribution on the set $\mathcal{V}_1$ defined above, and typical choices have been somewhat magical and lack an explanation.
In the Goemans-Williamson setting,  $A$ is an $n \times n$  PSD matrix, and the density $\nu$ on $\Omega$ they choose to express $A$ as a convex combination is as follows:  pick a vector $v \in \mathbb{R}^n$ from the normal distribution with covariance matrix $A$.
We show that this distribution is the maximum entropy distribution $\nu^\star$ (corresponding to $A$) on $\mathcal{V}_1$ with base measure induced by the Lebesgue measure on $\mathbb{R}^n$, thus giving an optimization characterization of this measure; see Corollary \ref{cor:GW}. 
The proof relies on strong duality  and a closed form expression for the dual objective integral on $\mathcal{V}_1$; see Theorem \ref{thm:strongduality}.

\paragraph{Quantum entropy.}
In quantum mechanics, 
a density matrix $\rho$ is a trace one complex $n \times n$ PSD matrix and describes the statistical state of a system. 
The extreme points in the set of density matrices are the pure states or $\mathcal{P}_1.$
von Neumann defined a notion of entropy \cite{von1955mathematical} of $\rho$ that is computed by first writing $\rho$ as a convex combination $\sum_{i=1}^n \lambda_i u_iu_i^*$, where $\{u_i\}_{i \in [n]}$ is an orthonormal basis for $\C^n,$ and then computing the negative Shannon entropy of the $\lambda_i$'s.
While the von Neumann entropy is a mathematically elegant notion, it was  vigorously argued in the 1970s that it does not capture the uncertainty in $\rho$ \cite{Band1976,Park1977,SlaterEntropy1}. 
In fact, von Neumann's way  to write $\rho$ as a convex combination of pure states can be viewed as ``the most terse'', or entropy minimizing one.
In the same papers, an alternative way to define entropy of a density matrix was suggested -- as the entropy of the entropy maximizing distribution with marginal $\rho$ -- and referred to as the barycentric quantum entropy.
Unlike the von Neumann entropy, that has a simple formula ($-\Tr \rho \log \rho$)), 
 the barycentric entropy did not have an efficient algorithm that could compute it.
Our algorithm to compute entropy maximizing distributions for $\mathcal{P}_1$ mentioned above  directly implies a polynomial time algorithm to compute the barycentric entropy of a density matrix (that is sufficiently in the interior) along with the probability density that achieves it; see Corollary \ref{cor:quantum}.

\paragraph{Entropic barrier function.}
Bubeck and Eldan in \cite{BubeckEldan} proved that the entropic barrier of a convex body $K \subseteq \R^d$ is a $(1+o(1))n$-self-concordant barrier on $K$.
Roughly speaking, this barrier function, for a point in $K$ is defined to be the optimal value of a dual maximum entropy optimization problem when $\Omega=K$ and the measure is the Lebesgue measure on $K$.
The computability of this barrier function for a point $K$ is not known in general.
One obstacle is to get a reasonable bound on the norm of the optimal dual solution. 
An almost direct consequence of Theorem \ref{thm:boundingbox} implies such a bound for points that are sufficiently in the interior of $K$; see Corollary \ref{cor:boundingbox_convex}.    

\section{Preliminaries}\label{sec:preliminaries}

\paragraph{Notation.}
Let $\C,\R,\R_+,\N$ denote the complex, real, nonnegative real, and natural numbers respectively.
For $k,n \in \N$, let $\C^{k \times n}$ and $\R^{k \times n}$ denote the sets of $k \times n$ complex and real matrices respectively.
A matrix $M \in \C^{n \times n}$ is said to be Hermitian if $A = A^*$ where $*$ denotes the conjugate transpose.
A Hermitian matrix $M$ is said to be PD (positive definite) and PSD (positive semidefinite) if its eigenvalues are positive and nonnegative respectively.
For an $n \times n$ matrix $X$, we define $\diag(X)$ to be the length-$n$ vector of the diagonal entries of $X$.
If $x$ is a vector, then we define $\diag(x)$ to be the diagonal matrix with entries the entries of $x$.
For any $k,n \in \N$, we equip the vector space $\C^{k \times n}$ with the Frobenius inner product $\langle Y, Z \rangle := \Tr(YZ^*)$.
We also denote $\|Y\| := \sqrt{\langle Y, Y \rangle}$.
Note that $\langle Y, Z \rangle \in \R$ whenever $Y,Z$ are Hermitian, so that the set of $n \times n$ Hermitian matrices is a real Hilbert space of dimension $n^2$.
Also $\langle Y, Z \rangle \geq 0$ whenever $Y,Z$ are PSD.
We further let $B_\epsilon(Y)$ denote the open $\epsilon$-ball centered at $Y$ in the space in which $Y$ lives (e.g., the $n \times n$ Hermitian matrices).
Finally, we let $\hull(S)$ denote the convex hull of a set $S$ in some ambient vector space.

\paragraph{Manifolds.}
In general, we let $\Omega$ be any smooth manifold that is embedded in a $d$-dimensional real Hilbert space $V$ with inner product $\langle \cdot, \cdot \rangle$.
Let $\mathcal{L}(X)=B$ denote the affine space in which $\hull(\Omega)$ is full dimensional, i.e.,  every element $X \in \hull(\Omega)$ satisfies the equation $\mathcal{L}(X)=B$.
The concrete manifolds we consider are collections of matrices with some  structure.
In particular, for fixed $n \in \N$, consider the following manifold within $\C^{n \times n}$.
An \emph{$n \times n$ rank-$k$ PSD projection} is a PSD matrix with $k$ eigenvalues equal to $1$ and the rest equal to $0$.
\[
  \textstyle  \mathcal{P}_k = \mathcal{P}_k(n) :=  \{n \times n \text{ rank-$k$ PSD projections}\}.
\]
Note that $\mathcal{P}_k$ is also a manifold within the space of $n \times n$ Hermitian matrices.\footnote{Note that $\mathcal{P}_k$ is homeomorphic to a Grassmannian, i.e., the manifold of $k$-dimensional subspaces within an $n$-dimensional space.
The homeomorphism is explicitly given as the map which sends a rank-$k$ PSD projection to the $k$-dimensional subspace given by its image.}
Other manifolds we consider are the complex unit sphere $S_\C^n \subset \C^n$ (which is related to $\mathcal{P}_1$), the manifold of all rank one matrices (not necessarily trace one): $ \mathcal{V}_1:= \{ vv^\top: v \in \mathbb{R}^n \},$ and a convex body $K \subset \mathbb{R}^n$.

We would also like to consider the convex hull of a given manifold $\Omega$.
To make sense of such a notion, we need to consider the manifold as being embedded in some ambient vector space.
This ambient space often the space of $n \times n$ Hermitian matrices in our examples.
In general, we refer to the elements of $\hull(\Omega)$ as \emph{marginals} or \emph{marginals matrices}.

\paragraph{Group actions.}
It is useful to understand the symmetries of some of the manifolds mentioned above in terms of groups that act on them.
Recall that an $n \times n$ unitary matrix is an invertible matrix $U$ for which $U^{-1} = U^*$, and an $n \times n$ orthogonal matrix is an invertible matrix $O$ for which $O^{-1} = O^\top$.
The unitary and orthogonal groups ($U(n)$ and $O(n)$) act on the manifolds discussed above as follows:
\begin{itemize}
\setlength\itemsep{0em}
    \item $U(n)$ acts on column vectors in $S_\C^n$ and on $\hull(S_\C^n)$ by left multiplication.
    \item $U(n)$ acts on $\mathcal{P}_k$ and on $\hull(\mathcal{P}_k)$ by conjugation.
    \item $O(n)$ acts on $\mathcal{V}_1$ and on $\hull(\mathcal{V}_1)$ by conjugation.
\end{itemize}
Note that the actions of $U(n)$ on $S_\C^n$ and on $\mathcal{P}_1$ are compatible in the sense that for $x \in S_\C^n$ and $U \in U(n)$, we have $(Ux)(Ux)^* = U(xx^*)U^*$ where $xx^* \in \mathcal{P}_1$.

\paragraph{Relative interior.}

The convex set $\hull(\Omega)$ is not necessarily full dimensional in the ambient Hilbert space.
To define a notion of interior for $\hull(\Omega)$, we restrict to the minimal affine subspace in which $\Omega$ lives (this is given explicitly by $\mathcal{L}(X) = B$ discussed above).
More generally, we make the following definition.
\begin{definition}[\bf Relative interior] \label{def:interior}
    Fix a convex subset $S$ in a vector space $V$, and let $V_{\mathcal{L},B}$ be the minimal affine subspace in which $S$ lives.
    We say that $Y \in V$ is in the \emph{$\eta$-interior} of $S$ (for $\eta > 0$) if
    \[
        B_\eta(Y) \cap V_{\mathcal{L},B} \subseteq S.
    \]
    We say that $Y$ is in the \emph{interior} of $S$ if there exists $\eta > 0$ such that $Y$ is in the $\eta$-interior of $S$.
\end{definition}
Here we usually consider $S = \mathcal{P}_k(n)$, and we will be interested in the case where $\eta \geq \frac{1}{\text{poly}(n)}$.

\paragraph{Measures and densities.}

Often, the manifolds $\Omega$ we  consider have some geometric structure (e.g.,  it is a manifold with a group action), and we  want to consider measures which  interact nicely with this structure.
To make sure this happens, we restrict to the class of measures which are given by continuous density functions on $\Omega$.
To make sense of this, we  need a natural base measure $\mu$ on $\Omega$ which corresponds to the density function $f(X) \equiv 1$.
(E.g., in the case of $\Omega = \C^n$ or $\Omega = \R^n$, the Lebesgue measure often plays this role.)
In particular, the support of $\mu$ should be equal to $\Omega$.

In the case of $\Omega = \mathcal{P}_k$, there is a canonical measure which is appropriately called the uniform measure:
we define $\mu_k$ be the unique unitarily invariant measure on $\mathcal{P}_k$, where $U(n)$ acts by conjugation (as discussed above).
Hence, equivalently (and more formally), we restrict to the class of measures on $\mathcal{P}_k$ which are absolutely continuous with respect to $\mu_k$.
We prove here the existence of $\mu_k$, a classical result.

\begin{proposition}[Existence of $\mu_k$] \label{prop:invariant_measure}
    There exists a distribution $\mu_k$ on $\mathcal{P}_k$ (which we call the uniform distribution). If $X$ is a random variable distributed according to $\mu_k$, then $X$ and $UXU^*$ have the same distribution for any unitary $U$.
\end{proposition}
\begin{proof}
    Pick random complex unit vectors $v_i \in \C^{n+1-i}$ for $i \in [k]$. Note that $v_1 \in \C^n$. Now, map $v_2$ into $v_1^\perp \cong \C^{n-1}$, map $v_3$ into $\{v_1,v_2\}^\perp \cong \C^{n-2}$, etc. to obtain a collection of $k$ orthogonal vectors in $\C^n$. Form an $k \times n$ matrix $P$ by letting the $v_i$ be the rows of $P$. Defining $X := P^*P \in \mathcal{P}_k$ gives a distribution $\mu_k$ on $\mathcal{P}_k$.
    
    For unitary invariance, note that this property holds for the choice of $v_1$ by construction. This can then be inductively applied to $v_2,\ldots,v_k$ by composing the given unitary with the appropriate projection.
\end{proof}

\noindent
We also consider the standard Lebesgue measure on $\mathbb{R}^n$ for convex bodies and its pushforward measure $\mu$ through the map $v \mapsto vv^\top$ on $\mathcal{V}_1$. 
Note that $S_\C^n$ also has a canonical unitarily invariant measure, usually called the Haar measure.
The pushforward of this measure through the map $v \mapsto vv^*$ yields the unitarily invariant measure $\mu_1$ on $\mathcal{P}_1$.

\paragraph{Integration/Counting oracle.}

We are interested in computing the following exponential integral for a given  $Y$ in our Hilbert space $V$.

\begin{definition}[\bf Exponential integrals] \label{def:exp_integral}
    Fix $n \in \N$ and let $\mu$ be a measure with support $\Omega$, a manifold embedded in the real Hilbert space $V$. We define the following function on an  input $Y \in V$:
     $$   \mathcal{E}(Y) = \mathcal{E}_\mu(Y) := \log \int_{\Omega} e^{-\langle Y, X \rangle} d\mu(X).$$
    Whenever $\mu = \mu_k$ and $\Omega = \mathcal{P}_k$, we use the following shorthand notation
     $   \mathcal{E}_k(Y).$
    We sometimes also refer to these integrals as \emph{exponential integrals}.
\end{definition}

\noindent
A \emph{strong integration/counting oracle} for $\Omega$ and $\mu$  outputs two quantities, given an element  $Y$ from the ambient Hilbert space $V$ of $\Omega$:
\begin{enumerate}
    \item $\mathcal{E}_\mu(Y)$
    \item the matrix $\nabla \mathcal{E}_\mu(Y)$, defined so that the following holds for any  $Z \in V$:
    \[
        \langle \nabla \mathcal{E}_\mu(Y), Z \rangle = \left.\frac{d}{dt} \mathcal{E}_\mu(Y+tZ)\right|_{t=0}.
    \]
\end{enumerate}
In the case of $\Omega=\mathcal{P}_k$, $Y$ and $Z$ are Hermitian.
Further, since the measure $\mu_k$ is unitarily invariant, we can assume that $Y$ is diagonal and expect the running time of the counting oracle should  depend polynomially on $n$ and the number of bits needed to represent $e^{-y_i}$ for any $i$, where $y_1,\ldots,y_n$ are the eigenvalues (diagonal elements) of $Y$.

As we will show, in the special case when $\Omega=\mathcal{V}_1$ and $\mu$ is the pushforward of the Lebesgue measure, we can compute the integral $\mathcal{E}_\mu(Y)$ exactly in time polynomial in the bit complexity of $Y$ due to a direct formula.
This happens because the measure $\mu$ is a product measure, which is not the case for $\mu_k$.

\section{The maximum entropy framework}\label{sec:maxentropy}

In this section we present our maximum entropy  convex program.
Fix a manifold $\Omega$ in a $d$-dimensional real Hilbert space with inner product $\langle \cdot, \cdot \rangle$, and let $\mathcal{L}(X)=B$ denote the corresponding affine space containing $\Omega$.
 Let $\mu$ be the base measure on $\Omega$ and $A$ in $\mathcal{K} := \hull(\Omega)$.
 Our goal is to  find a density function $\nu$ with marginal $A$ that minimizes the KL-divergence with respect to $\mu$.

\begin{figure*}
    \centering
    \[
    \hspace{0.4cm}
    \begin{array}[t]{c|c}
        \quad\quad
        \bf Primal
        \quad\quad\quad
        &
        ~
        \bf Dual
        \\

        \quad\quad
        \begin{aligned}[t]
        &\sup_\nu \left[-\int_{\Omega} \nu(X) \log\left(\nu(X)\right) d\mu(X)\right] \\
            &\textrm{subject to:} \\
                &\qquad \nu: \Omega \to \R_{\geq 0}, \text{ $\mu$-measurable} \\
                &\qquad \int_\Omega X \nu(X) d\mu(X) = A \\
                &\qquad \int_\Omega \nu(X) d\mu(X) = 1
        \end{aligned}
        \quad\quad\quad

        &

        ~
        \begin{aligned}[t]
        &\inf_Y F_A(Y) = \inf_Y \left[\langle Y, A \rangle + \log \int_\Omega e^{-\langle Y, X \rangle} d\mu(X)\right] \\
                &\textrm{subject to:} \\
                &\qquad \mathcal{L}(Y)=0
        \end{aligned}
    \end{array}
    \]

    \caption{Primal and dual maximum entropy convex programs for $A$ in the interior of $\hull(\Omega)$.}
    \label{fig:primal_dual}
\end{figure*}

We use the shorthand $\primal_\mu(A)$ (or $\primal_k(A)$ if $\mu = \mu_k$) to refer to this primal optimization program.
We  mainly consider the case of $\mu = \mu_k$ and $\Omega = \mathcal{P}_k$ or $\Omega=\mathcal{V}_1$ with $\mu$ the pushforward of Lebesgue measure.
In these cases $Y$ will comes from some subspace of the $n \times n$ Hermitian matrices.
Drawing from the intuition that these base measures are uniform over the manifold, and hence in some sense maximize entropy, we say the KL-divergence minimizing measure is entropy maximizing. 
However, we note that this framework is also applicable to other base measures,
in particular to the case when $\Omega$ is a convex body in $\mathbb{R}^d$ and $\mu$ is the Lebesgue measure.
The fact that the entropy integral (without the minus sign) is convex as a function of the density $\nu$ follows from the fact that this integral is precisely the KL divergence between the probability distribution corresponding to $\nu$ and the distribution $\mu$.
Convexity of the KL divergence for probability distributions is then a well-known  fact.

Efficiently solving this convex program directly is a priori impossible as the support of $\nu$ is infinite.
To find a succinct representation for the optimal $\nu^\star$, we turn to the dual program (see Section \ref{sect:dualform} for a derivation), which
gives us a nice representation of the max-entropy density function $\nu^\star$.
We  often use the shorthand $\dual_\mu(A)$ (or $\dual_k(A)$ if $\mu = \mu_k$) to refer to this program.

In the case of $\mathcal{P}_k$ with uniform measure $\mu_k$, the optimal solution to $\dual_k(A)$ is given by a Hermitian matrix $Y^\star$.
By strong duality (see Theorem \ref{thm:strongduality}), this in turn shows that the max-entropy density function $\nu^\star$ takes on a nice form:
\[
    \nu^\star(X) \propto e^{-\langle Y^\star, X \rangle}.
\]
As a note, in the case of $\Omega = \mathcal{P}_k$ this matrix $Y^\star$ is only unique up to a shift by a multiple of the identity matrix.
Issues arising from non-uniqueness can be handled by restricting to the minimal affine subspace in which $\hull(\mathcal{P}_k)$ lives, as referred to in the discussion surrounding Definition \ref{def:interior}.
However, as $A$ tends to the boundary of $\hull(\Omega)$, $Y^\star$ can be seen to tend to infinity as the support of the measure $\nu^\star$ tends to lower dimensions.

\section{Formal statement of our results}

\subsection{Mathematical and computational results}

Our first result shows that strong duality holds.

\begin{theorem}[\bf Strong duality]\label{thm:strongduality}
   Let  $\Omega$ be a manifold that is embedded in a $d$-dimensional real Hilbert space with an inner product $\langle \cdot, \cdot \rangle$, and let $\mu$ be a measure supported on $\Omega$.
     For any $A$ in the relative interior of the convex hull of $\Omega$,
     the optimal values of the primal and dual objective functions coincide, and the corresponding max-entropy distribution has density function of the following form for some $Y^\star$:
    \[
        \nu^\star(X) \propto e^{-\langle Y^\star, X \rangle}.
    \]
 \end{theorem}

\noindent
The proof of this result uses standard techniques and appears in the appendix (Sections \ref{sect:sd_proof} and \ref{sec:general_sd}).
This result applied to $\mathcal{P}_k$ and $\mu_k$ shows that optimizing $\dual_k(A)$ is in fact equivalent to optimizing $\primal_k(A)$, and therefore the max-entropy measure has the exponential form described above.

With strong duality in hand, we focus on the computability of the optimal matrix $Y^\star$  for the dual program.
To do this we use a version of the ellipsoid algorithm (see Theorem \ref{thm:ellipsoid} and the algorithm that follows), for which we need two things.

First, we need an upper bound on some norm of the dual optimal solution.
If $Y^\star$ is the optimal solution, then the number of iterations of the ellipsoid algorithm depends on $\log \|Y^\star \|$.
That said, it may seem that a bound depending on $e^{1/\eta}$, where $\eta$ is such that $B_\eta(A) \subset \mathrm{hull}(\Omega)$, is enough to achieve polynomial dependence on $\frac{1}{\eta}$.
However, this is not enough,  since the integral appearing in the dual is polynomially dependent on the number of bits needs to represent $e^{-y_i}$, where the $y_i$'s are the entries or eigenvalues of a given input $Y$.
Hence, we actually need polynomial dependence on $\frac{1}{\eta}$, which is achieved in our bounding box result below.
Note that this issue is not  surprising, as it crops up in exactly the same way in the discrete maximum entropy case (see \cite{SinghV14}).

We give here a bounding box result which is more general than we  need for the rank-$k$ projections case ($\Omega = \mathcal{P}_k$ and $\mu = \mu_k$).
It relies on a key ``balance'' property of the measures.
This notion extends important properties of the discrete uniform measure to continuous measures on manifolds and is one of the key notions we introduce.

\begin{definition}[\bf Balanced measure] \label{def:balanced}
   A measure $\mu$ is said to be \emph{balanced} if for any $\delta > 0$ and $X \in \Omega$, we have that at least $\exp(-\mathrm{poly}(\delta^{-1}, d))$ of the mass of $\mu$ is contained in the $\delta$-ball about $X$ (where $d$ is the dimension of the ambient space in which $\hull(\Omega)$ lives). 
\end{definition}
We see in Definition \ref{def:twoparam_interior} how this notion can be used to give a more refined notion of interior (beyond the $\eta$ parameter discussed above).
Conceptually, it allows us to give an measure-theoretic relaxation of the notion of a separating hyperplane.

\begin{theorem}[\bf Bounding box]\label{thm:boundingbox}
  Let    
   $\mu$ be a measure supported on a manifold $\Omega$ embedded in a $d$-dimensional real Hilbert space. 
    Suppose that $\mu$ is balanced, in the sense of Definition \ref{def:balanced}.
    Further, let $A$ be an element of the $\eta$-interior of the convex hull of $\Omega$.
    Then there is an optimal solution $Y^\star$ to the dual program such that:
    \[
        \| Y^\star \| \leq \poly(\eta^{-1}, d).
    \]
\end{theorem}

\noindent
Corollary \ref{cor:boundingbox_rank_k}  and Corollary \ref{cor:boundingbox_convex} give bounds for  rank-$k$ projections and convex bodies as corollaries.

\begin{remark} \label{rem:boundingbox_differences}
   Our bounding box result significantly generalizes  the discrete case (Theorem 2.7 in \cite{SinghV14}).
    Uniform distribution in the discrete case has atoms of uniformly strictly positive (at worst singly-exponentially small) mass at all points, and this implies a bound on optimal dual solutions.
    In the continuous case this is no longer true,  the notion of balance then fills the gap.
\end{remark}

\noindent
Second, at each step of the ellipsoid algorithm, we need to be able to evaluate the dual objective function and its gradient at given input $Y$.
The hardest part of such a computation comes in evaluating $\mathcal{E}_\mu$, the exponential integral portion of the objective function.
We show that if we have access to such an evaluation oracle, then under very general conditions, we can compute the maximum entropy distribution.

\begin{theorem}[Ellipsoid method-based general algorithm] \label{thm:main_algorithm}
    Let $\mu$ be a balanced measure with support on a manifold $\Omega$ embedded in a $d$-dimensional real Hilbert space. 
    Let the affine space in which $\Omega$ lies, $\mathcal{L}(X) = B$, be given as input $(\mathcal{L},B)$.
    Assume that $\Omega$ is 
    contained in a ball of radius $r$. There exists an algorithm that, given $A$ in the $\eta$-interior of $\hull(\Omega)$, any $\epsilon > 0$, and a strong counting/integration oracle for the exponential integral $\mathcal{E}_\mu(Y)$, returns $Y^\circ$ such that
    \[
        F_A(Y^\circ) \leq F_A(Y^\star) + \epsilon
    \]
    where $F_A$ is the objective function for the dual program $\dual_\mu(A)$, and $Y^\star$ is an optimum of the dual program. The running time of the algorithm is polynomial in $d$, $\eta^{-1}$, $\log(\epsilon^{-1})$, $\log(r)$, and the number of bits needed to represent $A$, $\mathcal{L}$, and $B$.
\end{theorem}

\noindent
Our next result says that in fact we have an efficient strong counting oracle for $\mathcal{E}_k$ on the domain $\mathcal{P}_k$ with measure $\mu_k$.

\begin{theorem}[Counting oracle]\label{thm:counting}
  There is an algorithm that, given $n \in \N$, 
  $k \in [n]$, an $n \times n$ real diagonal matrix $Y = \diag(y)$, and a $\delta > 0$, returns numbers $\bar{E}, \bar{G}$ such that 
    \begin{enumerate}
        \item $|\bar{E} - \mathcal{E}_k(Y)| \leq \delta$
        \item $|\bar{G} - \nabla \mathcal{E}_k(Y) | \leq \delta$, 
    \end{enumerate}  
    where $\mathcal{E}_k$ is the exponential integral defined above (and in Definition \ref{def:exp_integral}). The running time of the algorithm is polynomial in $n$, $\log(\frac{1}{\delta})$, and the number of bits needed to represent $e^{-y_i}$ for any $i \in [n]$.

\end{theorem}

\noindent
The proof of this theorem for $k=1$ is elementary but relies on the interesting connection between the complex unit sphere and the probability simplex. 
This connection also yields an exact sampling algorithm; see Proposition \ref{prop:rank_one_sampling}.
For $k>1$, the proof of the theorem above relies on the Harish-Chandra-Itzykson-Zuber formula \cite{HarishChandra1957}, \cite{IZ1980}; see Theorem \ref{thm:HCIZ}.

\begin{remark}
In the case of $\mathcal{V}_1$ with the pushforward of Lebesgue measure, there is an exact formula to compute the corresponding dual optimum for positive definite marginals $A$: $Y^\star = \frac{1}{2}A^{-1}$; 
  see Corollary \ref{cor:sdp_rounding}.
       Positive-definiteness of the input $Y$ is in fact required for the dual objective to be finite, which is in stark contrast with the $\mathcal{P}_k$ case where any Hermitian matrix is allowed.
    These points suggest a conceptual divide between the Lebesgue measure case and the rank-$k$ projections case.
We do not expect such a formula for $Y^\star$ in the case of $\mathcal{P}_1$ and, indeed, the lack of one has been one of the obstacles for efficient algorithms for quantum barycentric entropy and computing the normalizing constant of the matrix Bingham distribution.
\end{remark}

\begin{remark}
In this paper we primarily consider the best possible setting where the running time of the counting oracle depends logarithmically on the accuracy. 
We refer to such counting oracles as {\em exact}.
We note that our framework does allow for counting oracles where the dependence is polynomially in $1/\delta$.
\end{remark}

\begin{remark}
    Guler in \cite{guler1996barrier} studies the characteristic function of a convex cone.
    In our language, the characteristic function of a cone is the exponential integral $\mathcal{E}_K(y)$ with respect to the Lebesgue measure on the dual cone $K$:
    \[
        \mathcal{E}_K(y) = \log \int_K e^{-\langle y, x \rangle} dx.
    \]
    For the case of homogeneous convex cones, Guler gives a nice way to construct explicit formulas for the characteristic function.
    (A homogeneous cone is a cone $K$ such that for all $u,v \in K$ theres is a linear isomorphism of $K$ which maps $u$ to $v$.
    Orthants, Lorentz cones, and semidefinite cones are all homogeneous.
    See Sections 3 and 7 of \cite{guler1996barrier} for more details.)
    Given a fixed vector $e$ in the interior of $K$, any other vector $y$ in the interior of $K$, and an automorphism $A_y$ of $K$ mapping $e$ to $y$, the dual objective for $K$ can be written up to additive constant as:
    \[
        F_\theta(y) = \langle y, \theta \rangle - \log \int_K e^{-\langle y, x \rangle} dx = \langle y, \theta \rangle - \frac{1}{2} \log(\det(A_yA_y^\top)).
    \]
    Such an explicit formula gives a route to efficiently computing the dual objective function in this case.
\end{remark}

\noindent
The bounding box and counting oracle for $\mu_k$ and $\mathcal{P}_k$ then imply that the ellipsoid method-based algorithm from Theorem \ref{thm:main_algorithm} gives a polynomial time algorithm for approximately computing $Y^\star$, the optimum of the program $\dual_k(A)$.

\begin{corollary}[Ellipsoid method-based efficient algorithm for $\mathcal{P}_k$]\label{cor:maxentropy_algo}
    There exists an algorithm that, given $n \in \N$, $k \in [n]$, a trace-$k$ PD matrix $A$ in the $\eta$-interior of the convex hull of the set of $n \times n$ rank-$k$ PSD projection matrices (i.e., $\hull(\mathcal{P}_k)$), and an $\epsilon > 0$, returns a Hermitian matrix $Y^\circ$ such that
    \[
        F_A(Y^\circ) \leq F_A(Y^\star) + \epsilon,
    \]
    where $F_A$ is the dual objective function and $Y^\star$ is an optimal solution to the dual program $\dual_k(A)$. The running time of the algorithm is polynomial in $n$, $\frac{1}{\eta}$, $\log(\frac{1}{\epsilon})$, and the number of bits needed to represent $A$.
\end{corollary}

\noindent
We further discuss the closeness of the distributions associated to $Y^\circ$ and $Y^\star$ from the previous Corollary in Appendix \ref{sec:closeness}.

\begin{remark} \label{rem:boundingbox_polyeta}
    Notice that the dependence on $\frac{1}{\eta}$ means that we do not achieve a polynomial time algorithm for $A$ near the boundary of $\hull(\mathcal{P}_k)$.
    This dependence comes from the fact that the bounding box (Theorem \ref{thm:boundingbox}) is dependent on $\frac{1}{\eta}$.
    One may then naturally ask whether this bounding box dependence can be improved.
    It turns out that it cannot in this case, see Remark \ref{rem:bound_dependence}.
    Note that this differs from the discrete case, where in \cite{StraszakV19} the authors are able to remove this $\frac{1}{\eta}$ dependence under certain assumptions on the polytope.
\end{remark}

\subsection{Applications} \label{sec:applications}

\paragraph{Barycentric quantum entropy.}

In \cite{SlaterEntropy1}, Slater discusses the notion of barycentric quantum entropy of a density matrix, and compares it to that of von Neumann entropy.
His investigation of this notion was prompted by the work of Band and Park \cite{Band1976,Park1977}, who critiqued the use of von Neumann entropy as a good indicator of the uncertainty of the given density matrix.
In particular, they argue that a better notion of entropy would relate to distributions on all possible pure states, whereas the von Neumann entropy is derived from the discrete distribution on the pure states corresponding to eigenvectors of the matrix.
In response to this, Slater defines a notion of quantum entropy in terms of a max-entropy program on the set of all pure states.
He then goes on to show how one might determine the quantum entropy in a few specific cases.

\begin{definition}[Barycentric quantum entropy]
    Let $\rho$ be an $n \times n$ Hermitian density matrix (trace-1, positive semidefinite). Then the \emph{barycentric quantum entropy} of $\rho$ is defined  (in our notation) as:
$  
        H_b(\rho) := \inf_\nu \int_{\mathcal{P}_1} \nu(X) \log(\nu(X)) d\mu_1(X)
 $
    subject to
    \[
        \nu(X) \geq 0 ~\;  \forall X \in \mathcal{P}_1 \qquad \text{and} \qquad \int_{\mathcal{P}_1} \nu(X) d\mu_1(X) = 1 \qquad \text{and} \qquad \int_{\mathcal{P}_1} X \nu(X) d\mu_1(X) = \rho,
    \]
    where $\mathcal{P}_1$ denotes the set of pure states and $\mu_1$ denotes the unitarily invariant measure  on $\mathcal{P}_1.$
\end{definition}

\noindent
Our results for computing max-entropy measures on $\mathcal{P}_1$ immediately imply efficient computability of the barycentric quantum entropy for 
density matrices that are polynomially in the interior.

\begin{corollary}[Computability of  barycentric quantum entropy]\label{cor:quantum}
    There exists an algorithm that, given a Hermitian density matrix $\rho$ in the $\eta$-interior of the set of Hermitian density matrices and an $\epsilon > 0$, returns a number $\bar{H}$ such that $|\bar{H} - H_b(\rho)| < \epsilon$. The running time of the algorithm is polynomial in $n$, $\frac{1}{\eta}$, $\log(\frac{1}{\epsilon})$, and the number of bits needed to represent $\rho$.
\end{corollary}

\paragraph{Goemans-Williamson SDP rounding.}

In their seminal paper,  Goemans-Williamson \cite{GoemansW1995} gave a rounding scheme that gives a way to round a given PD matrix $A$ to a vector.
Their method goes by drawing a vector $v$ from a particular distribution on $\R^n$ based on the matrix $A$.

\begin{definition}[Goemans-Williamson measure]
    Given $n \in \N$ and a real positive definite $n \times n$ matrix $A$, the \emph{Goemans-Williamson measure} $\mu_{\mathrm{GW}}$ can be defined via a sampling process on $\R^n$ as follows.
    \begin{enumerate}
        \item Sample $g \in \R^n$ from the standard multivariate Gaussian distribution.
        \item Compute $v := Vg$ where $V$ is a square root of $A$, i.e., $A = VV^\top$.
        \item $v$ is a sample from $\mu_{\mathrm{GW}}$.
    \end{enumerate}
\end{definition}
It is then straightforward to compute the marginals matrix associated to this distribution as follows:
\[
    \mathbb{E}[vv^\top] = \int_{\R^n} (vv^\top) d\mu_{\mathrm{GW}}(v) = V\left[\int_{\R^n} gg^\top dg\right]V^\top = VV^\top = A.
\]
Thus, if we map $\mathbb{R}^n$ to $\mathcal{V}_1$ via $v \mapsto vv^\top$ and also pushforward the Lebesgue measure through this map, the above is precisely the marginal constraint in our max-entropy framework.
This observation implies  that the pushforward of the measure $\mu_{\mathrm{GW}}$ is a (strictly) feasible solution to the max-entropy primal program  on the domain $\mathcal{V}_1$ with the pushforward of the Lebesgue measure.
We show that it is also the optimal solution to the max-entropy program.

\begin{corollary}[Goemans-Williamson measure maximizes entropy]\label{cor:GW}
For any positive definite matrix $A$, let $\mu_{\mathrm{GW}}$ be the measure corresponding to the Goemans-Williamson  rounding scheme for $A$. Then the pushforward of $\mu_{\mathrm{GW}}$ to $\mathcal{V}_1$ is the max-entropy measure with marginals $A$ on $\mathcal{V}_1$ with respect to the pushforward of Lebesgue measure.
\end{corollary}

\paragraph{Entropic barrier function.}

Bubeck and Eldan in \cite{BubeckEldan} prove that the entropic barrier of a convex body $K \subseteq \R^d$ is a $(1+o(1))n$-self-concordant barrier on $K$, improving a seminal result of Nesterov and Nemirovski \cite{nesterov1994interior}.
In fact this gives the first explicit construction of a universal barrier for convex bodies with optimal self-concordance parameter.
\begin{definition}[Entropic barrier]
    Given a convex body $K \subseteq \R^d$, define the \emph{entropic barrier} for $K$ as the real-valued function on the interior of $K$ defined as:
    \[
        B_K(v) := \sup_{y \in \R^d} \left[\langle y, v \rangle - \log \int_K e^{\langle y, x \rangle} dx\right].
    \]
    Note that $-B_K(v)$ is precisely the maxium entropy dual program, up to negation of $y$ in the expression.
\end{definition}

\noindent
Open questions still remain about the efficient computability of the entropic barrier.
This is in particular true in the case where $K$ is a polytope, given as a membership oracle.
Towards this, the following is essentially a corollary to Theorem \ref{thm:boundingbox} (see Section \ref{sec:boundingbox_convex} for a full proof),
and can be used to efficiently compute the entropic barrier at points which are in the $\eta$-interior of $K$.

\begin{corollary}[Bounding box for convex bodies]\label{thm:boundingbox_convex}
    Let $\Omega \subset \R^d$ be a convex body contained in a ball of radius $R$.
    Further, let $A$ be an element of the $\eta$-interior of the convex hull of $\Omega$.
    Then there is an optimal solution $Y^\star$ to the dual program such that $\| Y^\star \| \leq \poly(\eta^{-1}, d, \log(R))$.
\end{corollary}

\noindent
Details of how this implies computability of the entropic barrier are omitted from this paper.

\section{Technical overview}

In this section, we give overviews of the proofs of the main  results of this paper
 and compare our techniques with those of previous work.
We start by describing the approach of \cite{SinghV14} in the case of discrete uniform measures $\mu$ with finite support $\Omega \subseteq \{0,1\}^d$.
In this case, the marginals vector $A$ of a measure $\nu$ on $\Omega$ is defined by setting $A_k$ to be the expected value of the $k$th entry of $x$ when $x$ is chosen according to $\nu$.
Note that the marginals vector $A$ is always an element of $\hull(\Omega)$.
The problem the authors of \cite{SinghV14} solve is described as follows: given a finite subset $\Omega$ and a desired marginals vector $A$  in the $\eta$-interior of $\hull(\Omega)$, compute the probability measure on $\Omega$ with marginals $A$ which maximizes entropy.

They consider the dual formulation
\[
    \inf_{y \in \R^d} \; \;  F_A(y) :=  \langle y, A \rangle + \log \sum_{x \in \Omega} e^{-\langle x, y \rangle},
\]
which gives rise to measures on $\Omega$ of the following succinct form for some real vector $y^\star$:
\[
    \nu(x) \propto e^{-\langle x, y^\star \rangle}.
\]
By strong duality $\nu=\nu^\star$ is the entropy maximizing measure, and they then use the ellipsoid method to approximate $y^\star$.

We generalize their approach to continuous measures $\mu$ on continuous domains $\Omega$.
For the most part, the ellipsoid algorithm can be applied in the same way as in the discrete case once we have the three main pieces in hand: (1) strong duality, (2) a bound on $Y^\star$, and (3) the strong counting oracle.
Even in the continuous case, one can show that strong duality holds via a certain Slater-type condition (see Sections \ref{sect:sd_proof} and \ref{sec:general_sd}).
What makes the passage from the discrete case to the continuous case much more interesting and nontrivial is proving the remaining two main pieces.

\subsection{Proof overview: bounding box}

The goal of this section is to explain the proofs of the main bounding box result and its corollaries.
We first describe the approach of the discrete $\mu$ case discussed above.
Note that for $B \in \hull(\Omega)$, there exists some $X_0 \in \Omega$ such that
\[
    \langle -Y^\star, X_0-B \rangle \geq 0,
\]
since every closed half-space containing $B$ contains an extreme point $X_0 \in \hull(\Omega)$.
If $A$ is in the $\eta$-interior of $\hull(\Omega)$, we can choose $B = A - \eta \frac{Y^\star}{\|Y^\star\|}$ to get:
\[
    \langle -Y^\star, X_0-A \rangle \geq \eta \|Y^\star\|.
\]
Because $\mu$ is a discrete uniform measure, we have $\mu(\{X_0\}) = |\Omega|^{-1}$.
This implies a bound on $Y^\star$ as follows, via the dual objective function $F_A(Y)$:
\[
    0 = F_A(0) \geq F_A(Y^\star) = \log \int e^{\langle -Y^\star, X-A \rangle} d\mu(X) \geq \log\left(e^{\eta\|Y^\star\|} \cdot |\Omega|^{-1}\right) \implies \|Y^\star\| \leq \frac{\log|\Omega|}{\eta}.
\]
The lower bound on $F_A(Y^\star)$ above follows from restricting the integral (which is a sum in the discrete case) to the single point $X_0$.
This demonstrates exactly why this argument fails in the continuous case, because in that case we have $\mu(\{X\}) = 0$ for all $X \in \Omega$.

This is the first difficulty we must overcome.
We need a way to restrict the dual objective integral to a region of $\Omega$ which has positive mass, emulating the role of atoms in the discrete case.

We introduce a two-parameter interior for the measure $\mu$.
We say that $A$ is in the $(\eta,\delta)$-interior of $\mu$ if every half-space intersecting the $\eta$-ball about $A$ contains at least $\delta$ mass of $\mu$ (Definition \ref{def:twoparam_interior}).
Instead of restricting the dual integral to a single point of $\Omega$, we restrict it to the appropriate $\delta$-mass to obtain a bound on $\|Y^\star\|$:
\[
    0 \geq \log \int e^{\langle -Y^\star, X-A \rangle} d\mu(X) \geq \log\left(e^{\eta\|Y^\star\|} \cdot \delta\right) \implies \|Y^\star\| \leq \frac{1}{\eta} \log\frac{1}{\delta}.
\]
We explain this formally in Lemma \ref{lem:twoparam_bounding}.

This leads to the second difficulty.
Our bounding box theorem only refers to the $\eta$ parameter, and so we need a way to handle or control $\delta$ in terms of $\eta$ and $d$.

Here is where the key \emph{balance} property comes into play.
We say that a measure $\mu$ is balanced if for all $\epsilon > 0$ and $X \in \Omega$, the $\epsilon$-ball about $X$ contains $\exp(-\poly(\epsilon^{-1}, d))$ of the mass of $\mu$ (Definition \ref{def:balanced}).
This links the two interiority parameters: from any point of the $\epsilon$-interior of $\hull(\Omega)$, there will be at least $\exp(-\poly(\epsilon^{-1}, d))$ mass in the direction of any $X \in \Omega$ on the boundary.

The crucial feature of the balance property is then how this linking of the parameters allows one to transfer between them.
Specifically for a balanced measure, the $\eta$-interior of $\hull(\Omega)$ is contained in the $(\frac{\eta}{2},\exp(-\poly(\frac{\eta}{2}, d)))$-interior of $\mu$.
To see this, let $A$ be in the $\eta$-interior of $\hull(\Omega)$.
Hence, any half space which intersects the $\frac{\eta}{2}$-ball about $A$ contains another $\frac{\eta}{2}$-ball in $\hull(\Omega)$.
By translating this ball toward a point of $\Omega$, we can assume that the half-space contains an $\frac{\eta}{2}$-ball about a point of $\Omega$.
Since $\mu$ is balanced, this implies $A$ is in the $(\frac{\eta}{2}, \exp(-\poly(\frac{\eta}{2}, d)))$-interior of $\mu$.

At this point, the rest of the proof of Theorem \ref{thm:boundingbox} is straightforward.
For balanced $\mu$ and $A$ in the $\eta$-interior of $\hull(\Omega)$, we actually have that $A$ is in the $(\frac{\eta}{2}, \exp(-\poly(\frac{\eta}{2}, d)))$-interior of $\mu$.
The two parameter bound described above then implies $\|Y^\star\| \leq \poly(\frac{1}{\eta}, d)$.

To obtain bounding boxes for $\mu_k$ on $\mathcal{P}_k$, $n \times n$ rank $k$ projections, (Corollary \ref{cor:boundingbox_rank_k}) and to uniform measures on convex bodies (Corollary \ref{cor:boundingbox_convex}), we then demonstrate balance properties.
In the case of $\mu_k$, $\mathcal{P}_k \subset B_{\sqrt{k}}(0)$ can be covered by at most $\exp(\poly(\log\delta^{-1}, n))$ balls of radius $\delta$ for any $\delta > 0$, morally because:
\[
    \frac{\mathrm{vol}(B_{\sqrt{k}})}{\mathrm{vol}(B_\delta)} = \frac{(\pi \sqrt{k})^n / n!}{(\pi \delta)^n / n!} = \left(\frac{\sqrt{k}}{\delta}\right)^n = \exp(\poly(\log\delta^{-1}, n)).
\]
Therefore a $\delta$-ball about \emph{some} point of $\mathcal{P}_k$ must contain at least $\exp(-\poly(\log\delta^{-1}, n))$ of the mass of $\mu_k$, and unitary invariance then implies that this is actually true for \emph{all} points of $\mathcal{P}_k$.

For uniform measures $\mu$ on convex bodies $K$ contained in a ball of radius $R$, we prove the bounding box using similar arguments as follows.
By the volume ratio computation above, every $\delta$-ball contained in $K$ contains at least $(\frac{\delta}{R})^d$ of the mass of $\mu$.
Therefore every $A$ in the $\eta$-interior of $\hull(\Omega)$ is also in the $(\frac{\eta}{2}, (\frac{\eta}{2R})^d)$-interior of $\mu$, since every half-space intersecting the $\frac{\eta}{2}$-ball about $A$ contains another $\frac{\eta}{2}$-ball in $K$.
The bounding box then follows from the two-parameter bound discussed above (Lemma \ref{lem:twoparam_bounding}).

\subsection{Proof overview: counting oracle for $\mathcal{P}_1$ and $\mathcal{V}_1$}

The goal of this section is to explain why we can efficiently evaluate and compute the gradient of
\[
    \mathcal{E}_\mu(Y) = \log \int_\Omega e^{-\langle Y, X \rangle} d\mu(X)
\]
in the case of $\Omega = \mathcal{P}_1$ and $\Omega = \mathcal{V}_1$

First consider the case of $\Omega = \mathcal{V}_1$, where $\mu$ is the pushforward of the Lebesgue measure through $x \mapsto xx^\top$.
In this case we have a very explicit formula whenever $Y$ is positive definite:
\[
    \mathcal{E}_\mu(Y) = \log \int e^{-\langle Y, X \rangle} d\mu(X) = \frac{n}{2}\log(\pi) - \frac{1}{2}\log\det(Y).
\]
Since $\mu$ is the pushforward of the Lebesgue measure through $x \mapsto xx^\top$, this expression follows from the following classical Gaussian integral formula:
\begin{equation}\label{eq:GW_TO}
    \int_{\mathcal{V}_1} e^{-\langle Y, X \rangle} d\mu(X) = \int_{\R^n} e^{-x^\top Y x} dx = \sqrt{\det(\pi Y)}.
\end{equation}
This is demonstrated formally in Proposition \ref{prop:lebesgue_formula}.
We show how leads to our optimality characterization of the Goemans-Williamson measure at the end of this section.

The above Gaussian formula for $\mathcal{V}_1$ suggests a natural approach for computing $\mathcal{E}_1$ on $\mathcal{P}_1$.
Allowing complex Hermitian matrices, note that $\mathcal{P}_1$ is the set of norm-1 elements of $\mathcal{V}_1$.
Hence, we ``integrate out'' the norm of the elements of $\mathcal{V}_1$, in an attempt to obtain a similar formula for $\mathcal{P}_1$.
We do this via a standard change of variables (equalities are up to scalar):
\[
    \int_{\mathcal{V}_1} e^{-\langle Y, X \rangle} d\mu(X) = \int_{\mathcal{P}_1} \int_0^\infty e^{-\langle Y, r^2X \rangle} r^{2n-1} dr d\mu_1(X) = \int_{\mathcal{P}_1} \langle Y, X \rangle^{-n} d\mu_1(X) \neq \mathcal{E}_1(Y).
\]
This shows that this approach fails: that is, integrating out the norm does not provide us a formula for $\mathcal{E}_1(Y)$ (for more discussion see Section \ref{sec:GWsphere}).

This demonstrates the first difficulty for constructing a counting oracle for $\mathcal{P}_1$.
Normalizing the max-entropy measure on $\mathcal{V}_1$ as above yields a measure on $\mathcal{P}_1$ which is \emph{not} a max-entropy measure.
Max-entropy measures on $\mathcal{P}_1$ an $\mathcal{V}_1$ are therefore fundamentally different objects, and thus constructing the associated counting oracles  requires  different techniques.
In particular the well-known Gaussian integral formulas cannot help us in the case of $\mathcal{P}_1$.

The remarkable fact is then that max-entropy measures on $\mathcal{P}_1$ \emph{can} be translated into max-entropy measures on a very simple polytope: the standard simplex in $\R^n$.
We have the following equality for real $Y = \diag(y)$, where $m$ is the Lebesgue measure on the simplex $\Delta_1:=\{p \in \mathbb{R}_+^n: \sum_{i=1}^n p_i=1\}$:
\[
    \int_{\mathcal{P}_1} e^{-\langle Y, X \rangle} d\mu_1(X) = \int_{\Delta_1} e^{-\langle y, x \rangle} dm(x).
\]
Put another way, max-entropy measures on $\mathcal{P}_1$, a nonconvex manifold, correspond to max-entropy measures on $\Delta_1$, a convex polytope.
To see this, first note the following for any $m_1,\ldots,m_n$. The first equality is the Bombieri inner product formula (Lemma \ref{lem:bombieri}), and the second inequality is a basic induction after a change of variables:
\[
    \int_{\mathcal{P}_1} X_{11}^{m_1} \cdots X_{nn}^{m_n} d\mu_1(X) = \frac{m_1! \cdots m_n! (n-1)!}{(m_1+\cdots+m_n+n-1)!} = \int_{\Delta_1} x_1^{m_1} \cdots x_n^{m_n} dm(x).
\]
The exponential equality then follows from limiting, since $\mathcal{P}_1$ and $\Delta_1$ are compact and since $e^{-\langle Y, X \rangle}$ and $e^{-\langle y, x \rangle}$ are limits of polynomials.

This argument also implies the more general fact: that $m$ is the pushforward of $\mu_1$ through the map $\phi: X \mapsto \diag(X)$:
\[
    \int_{\mathcal{P}_1} f(\phi(X)) d\mu_1(X) = \int_{\Delta_1} f(x) dm(x).
\]
This transfer to the simplex now leads to an explicit computation for $\mathcal{E}_1(Y)$ when $Y = \diag(y)$.
(Considering diagonal $Y$ is actually without loss of generality, see the discussion in Section \ref{sec:counting_oracle}.)
By making a change of variables, the simplex integral is an iterated convolution:
\[
    \frac{1}{(n-1)!} \int_{\mathcal{P}_1} e^{-\langle Y, X \rangle} d\mu_1(X) = \int_0^1 \int_0^{1-x_1} \cdots \int_0^{1-x_1-\cdots-x_{n-2}} e^{-\langle y, x \rangle} dx = \left.(e^{-y_1 t} * \cdots * e^{-y_n t})\right|_{t=1}.
\]
This is stated formally in Lemma \ref{lem:laplace_integral}.
Applying the Laplace transform $\mathcal{L}$ converts this convolution into a partial fraction decomposition problem for distinct values of $y_i$:
\[
    \left.(e^{-y_1 t} * \cdots * e^{-y_n t})\right|_{t=1} = \mathcal{L}^{-1}\left[\frac{1}{\prod_i (s+y_i)}\right](1) = \mathcal{L}^{-1}\left[\sum_i \frac{c_i}{s+y_i}\right](1) = \sum_i c_i e^{-y_i}.
\]
Computing the values of $c_i$ via a standard partial fractions formula gives:
\[
    \frac{1}{(n-1)!} \int_{\mathcal{P}_1} e^{-\langle Y, X \rangle} d\mu_1(X) = \sum_{i=1}^n \frac{e^{-y_i}}{\prod_{j \neq i} (y_j-y_i)} = \frac{\det(M(-y))}{\prod_{i < j} (y_j-y_i)}.
\]
This is stated formally in Proposition \ref{prop:evaluation_formula}.
Here $M(-y)$ is a Vandermonde-like matrix which arises when forming the common denominator of the last expression, given (for the case of distinct $y_i$s) as follows:
\[
    M(-y) := \left[
        \begin{matrix}
            1 & 1 & \cdots & 1 \\
            (-y_1) & (-y_2) & \cdots & (-y_n) \\
            (-y_1)^2 & (-y_2)^2 & \cdots & (-y_n)^2 \\
            \vdots & \vdots & \ddots & \vdots \\
            (-y_1)^{n-2} & (-y_2)^{n-2} & \cdots & (-y_n)^{n-2} \\
            e^{-y_1} & e^{-y_2} & \cdots & e^{-y_n} \\
        \end{matrix}
    \right].
    \]
We define this matrix formally in Definition \ref{def:eval_matrix}.

This brings us to the second difficulty for constructing a counting oracle for $\mathcal{P}_1$.
When the values of $y_i$ are not distinct, then the denominator vanishes and this formula cannot be used.
Even though $\mathcal{E}_1$ is continuous, this could still be a major problem: if for example the gradient of $\mathcal{E}_1(Y)$ becomes large as $y_i$ approaches $y_j$, then computing $\mathcal{E}_1(Y)$ could become computationally infeasible.

To handle this difficulty, we take limits by successively applying L'Hopital's rule.
One iteration for $y_1 = y_2$ goes as follows:
\[
    \lim_{y_2 \to y_1} \frac{\det(M(-y))}{\prod_{i < j} (y_j-y_i)} = \left.\frac{\partial_{y_2}\det(M(-y))}{\partial_{y_2}\prod_{i < j} (y_i-y_j)}\right|_{y_2 = y_1} = \frac{\det(M'(-y))}{\prod_{2 < i}(y_i-y_1)^2 \prod_{2 < i < j} (y_i-y_j)}.
\]
The key observation here is the fact that the numerator is still a determinant, due to the fact that only one column of $M(-y)$ depends on $y_i$ for all $i$.
Applying L'Hopital's rule as many times as is necessary leads to the following, where $\lambda_i$ represent the distinct values of $y$ with multiplicities $m_i$:
\[
    \exp(\mathcal{E}_1(Y)) = \int_{\mathcal{P}_1} e^{-\langle Y, X \rangle} d\mu_1(X) = (n-1)! \frac{\det(M(-\lambda))}{\prod_{i<j} (\lambda_i-\lambda_j)^{m_im_j}}.
\]
Note that $M(-\lambda)$ is a matrix similar to $M(-y)$ above which handles the non-distinctness (we unify the notation of these matrices in Definition \ref{def:eval_matrix}).
A similar expression for the gradient is achieved using the same techniques, and so we state it here without further detail:
\[
    (\nabla \mathcal{E}_1(Y))_l = -\sum_{i \neq p} \frac{m_i}{\lambda_p - \lambda_i} - \frac{\det(M_p(-\lambda))}{\det(M(-\lambda))}.
\]
$M_p(-\lambda)$ is another, similar Vandermonde-like matrix, see Proposition \ref{prop:gradient_formula} and Definition \ref{def:grad_matrix}.

Since the entries of $M(-\lambda)$ and $M_p(-\lambda)$ have bit complexity polynomial in $n$ and the bit complexity of $e^{-y_i}$, their determinants have the same bit complexity.
Therefore these formulas, for $\mathcal{E}_1(Y)$ and its gradient, lead to an efficient counting oracle for $\mathcal{P}_1$.

\paragraph{The optimality of Goemans-Williamson measure.} As a consequence of Equation \eqref{eq:GW_TO}, we now show briefly how this formula is used to prove that the Goemans-Williamson measure $\mu_{\mathrm{GW}}$ with respect to a real symmetric positive definite matrix $A$ is a max-entropy measure on $\mathcal{V}_1$.
For $A = VV^\top$, the measure $\mu_{\mathrm{GW}}$ is defined to be distributed according to $xx^\top := (Vg)(Vg)^\top$ where $g$ is a standard Gaussian in $\R^n$.
By the change of variables formula, $xx^\top$ is distributed as follows on $\R^n$:
\[
    xx^\top \sim e^{-\frac{1}{2}\|V^{-1}x\|^2} \cdot \det(V^{-1}) d\mu(xx^\top) \propto e^{-\langle \frac{1}{2}A^{-1}, xx^\top \rangle} d\mu(xx^\top).
\]
We state this formally in Proposition \ref{prop:GW_density}.
To prove that this is a max-entropy measure, we determine the critical point of the dual objective with respect to real symmetric positive definite $Y$:
\[
    0 = \nabla(\langle Y, A \rangle + \mathcal{E}_\mu(Y)) = A - \frac{1}{2}Y^{-1} \implies Y^\star = \frac{1}{2}A^{-1}.
\]
Therefore, the Goemans-Williamson measure $\mu_{\mathrm{GW}}(X) \propto e^{-\langle \frac{1}{2}A^{-1}, X \rangle} d\mu(X)$ is the max-entropy measure on $\mathcal{V}_1$ with respect to $A$.

\paragraph{Proof overview: sampling for $\mathcal{P}_1$.}

We now discuss how to sample from max-entropy distributions on $\mathcal{P}_1$.
Our main algorithm (Theorem \ref{thm:main_algorithm}) gives an efficient oracle for approximating the max-entropy density function:
\[
    \nu(X) \propto e^{-\langle Y^\star, X \rangle}.
\]
The main problem is that it is not at all clear how to use such a density function to sample from a manifold.

We avoid this difficulty by transferring the problem of sampling to the simplex $\Delta_1$ for $Y^\star = \diag(y^\star)$, using the following fact discussed in the previous section:
\[
    \frac{1}{(n-1)!} \int_{\mathcal{P}_1} e^{-\langle Y^\star, X \rangle} d\mu_1(X) = \int_0^1 \int_0^{1-x_1} \cdots \int_0^{1-x_1-\cdots-x_{n-2}} e^{-\langle y^\star, x \rangle} dx.
\]
The sampling process for $\mathcal{P}_1$ then occurs in two parts.

First, we sample from the max-entropy distribution on the simplex, one coordinate at a time.
We use the right-hand side of the above expression to compute the cumulative density function (CDF) for each coordinate, conditioned on the previously sampled coordinates.
Formulas and computations for these conditioned CDFs are very similar to that of the counting oracle, and hence we omit them here (see Corollary \ref{cor:cdf_formula_rank_one}).

Once we have a sample $x$ on the simplex, we need to convert it into a sample on $\mathcal{P}_1$ by considering its inverse image under the map $\phi: X \mapsto \diag(X)$.
The difficulty that now arises is the fact that there are many elements of $\mathcal{P}_1$ which map to the same simplex element under $\phi$.

Fortunately, there is a principled way to select from these possibilities.
The fiber $\phi^{-1}(x)$ is an orbit of the action of diagonal unitary matrices on $\mathcal{P}_1$ by conjugation. 
Since $Y^\star$ is diagonal, this implies the max-entropy measure $\nu(X)$ is uniform when restricted to $\phi^{-1}(x)$.
Given $x$, we then sample $X$ from $\phi^{-1}(x)$ by picking an arbitrary $X_0 \in \phi^{-1}(x)$ and conjugating by a uniformly random diagonal unitary matrix.

Hence, to sample $X$ from $\mathcal{P}_1$ we (1) sample $x$ from the simplex, and then (2) sample $X$ uniformly from $\phi^{-1}(x)$.
This samples $X$ from the correct measure due to the disintegration theorem, which says the following for any $f$:
\[
    \int_{\mathcal{P}_1} f(X) d\mu_1(X) = \int_{\Delta_1} \int_{\phi^{-1}(x)} f(X) d\mu_{\phi^{-1}(x)}(X) dm(x).
\]
That is, the measure $\mu_1$ can be split into measures on $\Delta_1$ and on the fibers $\phi^{-1}(x)$ (see Proposition \ref{prop:disintegration}).

Therefore, the above sampling process efficiently samples the max-entropy measure on $\mathcal{P}_1$ with density $\nu(X)$.

\subsection{Proof overview: extending the counting oracle for $\mathcal{P}_1$ to $\mathcal{P}_k$}

For the case of $\mathcal{P}_k$ and $\mu_k$, we want to generalize the formulas of the $k=1$ case.
To do this, we make use of the famous Harish-Chandra-Itzykson-Zuber formula (Theorem \ref{thm:HCIZ}) for integrals over the Haar measure of the unitary group $U(n)$.
It is stated as follows for Hermitian $Y,B$ with {\em distinct} eigenvalues $y_i,\beta_i$:
\[
    \int_{U(n)} e^{-\langle Y, UBU^* \rangle} dU = \left(\prod_{p=1}^{n-1} p!\right) \frac{\det([e^{-y_i\beta_j}]_{1 \leq i,j \leq n})}{\prod_{i < j} (y_i - y_j)(\beta_j - \beta_i)}.
\]
For $B = \diag(1,\ldots,1,0,\ldots,0)$ with $k$ 1s and $n-k$ 0s, notice that $\mathcal{P}_k = \{UBU^* ~:~ U \in U(n)\}$.
This leads to the following:
\[
    \exp(\mathcal{E}_k(Y)) = \int_{\mathcal{P}_k} e^{-\langle Y, X \rangle} d\mu_k(X) = \int_{U(n)} e^{-\langle Y, UBU^* \rangle} dU.
\]
To handle the issue of the denominator vanishing, and to compute the gradient, we apply all the same techniques which were required for the $k=1$ case (see Corollaries \ref{cor:dual_integral_k} and \ref{cor:gradient_formula_k}).
These formulas end up having the right bit complexity, and so they immediately imply an efficient strong counting oracle for $\mathcal{P}_k$.

Unlike in the case of $k=1$, the problem of sampling in the case of $k>1$ is more difficult as the image of $\mathcal{P}_k$ under the map $\phi: X \mapsto \diag(X)$ is much more complicated.
Thus we leave as an open problem the question of sampling from the associated maximum entropy distributions in the case of $\mathcal{P}_k$ for $k>1$.

\section{Bounding box} \label{sec:boundingbox}

In this section, we prove the general bounding box result (Theorem \ref{thm:boundingbox}).
With this, we then specialize to the cases of rank-$k$ projections and convex bodies.

\subsection{General bounding box}

In what follows we will discuss ``interiors'' of a probability distribution $\mu$ given by two parameters, $(\eta,\delta)$.
The $\eta$ parameter will control how far we are from the boundary, and the $\delta$ parameter will control how well-distributed $\mu$ is on its support.
At the end of the day, we will prove that for nice situations one only needs to consider the $\eta$ parameter (as in the bounding box result of \cite{SinghV14}).

We now define the two-parameter interior.
In what follows, we will let $V_\mathcal{L}$ be the vector subspace given by $\mathcal{L}(X) = 0$, where $\mathcal{L}(X) = B$ is the maximal set of linearly independent equality constraints for $\Omega$.
More informally, $V_\mathcal{L}$ is the vector space corresponding to the minimal affine space in which $\mathcal{K} = \hull(\Omega)$ lives (i.e., translate the affine space so that $0 \in V_\mathcal{L}$).
The fact that $\mathcal{L}(X) = B$ is a maximal linearly independent set means that the optimal solution to the dual program is unique when restricted to $V_\mathcal{L}$.
(Existence follows from Lemmas \ref{lem:linear_constraints} and \ref{lem:dual_optimum}.)
We discuss this further in Section \ref{sec:computing_max_entropy}.

\begin{definition} \label{def:delta_param}
    We define the $(0,\delta)$-interior of $\mu$ to be the set of all $A \in \mathcal{K}$ such that for all $Y \in V_\mathcal{L}$ we have:
    \[
        \mu(\{X \in \Omega ~|~ \langle X-A, Y \rangle \geq 0\}) > \delta.
    \]
\end{definition}

\noindent
Morally, this says that every closed half-space containing $A$ contains more than $\delta$ of the mass of $\mu$.
Note that this is not always an open set (which is perhaps a bit odd for something called the ``interior'', but this will be our convention).

\begin{definition}[Two-parameter interior] \label{def:twoparam_interior}
    We define the $(\eta,\delta)$-interior of $\mu$ to be the set of all $A \in \mathcal{K}$ such that the ball of radius $\eta$ about $A$ is contained in the $(0,\delta)$-interior of $\mu$.
    Note that this is not necessarily an open set.
\end{definition}

\noindent
The next lemma is then precisely how to combine the two parameters to get a bounding box for the optimal solution to the dual program.

\begin{lemma}[Two-parameter bounding box] \label{lem:twoparam_bounding}
    Given $A \in \mathcal{K}$, let $Y^\star \in V_\mathcal{L}$ be the optimal solution to the dual program.
    Recall the dual objective:
    \[
        \inf_Y F_A(Y) = \inf_Y \log \int_\Omega e^{-\langle Y, X-A \rangle} d\mu(X).
    \]
    If $A$ is in the $(\eta,\delta)$-interior of $\mu$, then $\|Y^\star\| \leq \frac{1}{\eta} \log\left(\frac{1}{\delta}\right)$.
\end{lemma}
\begin{proof}
    By definition, we have that $A - \eta \cdot \frac{Y^\star}{\|Y^\star\|}$ is in the $(0,\delta)$-interior of $\mathcal{K}$.
    Therefore:
    \[
        \delta \leq \mu(\{X \in \Omega ~|~ \langle X-(A - \eta \cdot Y^\star/\|Y^\star\|), -Y^\star \rangle \geq 0\}) = \mu(\{X \in \Omega ~|~ \langle X-A, -Y^\star \rangle \geq \eta \cdot \|Y^\star\|\}).
    \]
    This gives the bound:
    \[
        \log \int e^{\langle -Y^\star, X-A \rangle} d\mu(X) \geq \log\left(\delta \cdot e^{\eta \cdot \|Y^\star\|}\right) = \log(\delta) + \eta \cdot \|Y^\star\|.
    \]
    On the other hand, plugging in $Y = 0$ gives an upper bound on the optimal value of the above dual program:
    \[
        0 \geq \log \int e^{\langle -Y^\star, X-A \rangle} d\mu(X) \geq \log(\delta) + \eta \cdot \|Y^\star\|.
    \]
    Rearranging this gives the result.
\end{proof}

This gives us a good way of bounding solutions corresponding to interior points of $\mathcal{K}$.
In general however, trying to get a bound on the $\delta$ parameter of the interior is much more difficult than that of the $\eta$ parameter.
To deal with this we define a property of $\mu$ which allows us to only have to consider the $\eta$ parameter.

\begin{definition}[$\delta$-balanced measure] \label{def:balanced_in_section}
    We say that $\mu$ is \emph{$\delta$-balanced} if for any $X \in \Omega$, we have that at least $\exp(-\mathrm{poly}(\delta^{-1}, d))$ of the mass of $\mu$ is contained in the $\delta$-ball about $X$ (where $d$ is the dimension of $\mathcal{K}$). If $f$ is the polynomial in the exponent (i.e., $\exp(-f(\delta^{-1}, d))$), then we say that $\mu$ is \emph{$\delta$-balanced with bound $f$}.
\end{definition}

We now prove the main bounding box theorem for such balanced measures.
We then use this to obtain a bounding box for rank-$k$ projections and for convex bodies in the following sections.

\begin{theorem}[Bounding box for balanced measures] \label{thm:bounding-main}
    Suppose $\mu$ is $\frac{\eta}{2}$-balanced with bound $f$.
    If $A$ is in the $(\eta,0)$-interior of $\mu$ and $Y^\star \in V_\mathcal{L}$ is the optimal solution to the corresponding dual program, then $\|Y^\star\| \leq 2\eta^{-1} \cdot f(2\eta^{-1}, d) = \mathrm{poly}(\eta^{-1}, d)$.
\end{theorem}
\begin{proof}
    We first show that the $(\frac{\eta}{2},0)$-interior of $\mu$ is contained in the $(0,\exp(-f(\frac{2}{\eta}, d)))$-interior of $\mu$.
    To see this, let $A_0$ be some element of the $(\frac{\eta}{2},0)$-interior of $\mu$.
    Then any closed half-space containing $A_0$ also contains an $\frac{\eta}{2}$-ball about some $X \in \Omega$.
    That is, for every $Y \in V_\mathcal{L}$ there exists $X$ such that:
    \[
        B_{\eta/2}(X) \subseteq \{Z \in \Omega ~|~ \langle Z - A_0, Y \rangle \geq 0\}.
    \]
    Since $\mu$ is $\frac{\eta}{2}$-balanced, we have that $\exp(-f(\frac{2}{\eta}, d))$ of the mass of $\mu$ is contained in the $\frac{\eta}{2}$-ball about $X$.
    This implies:
    \[
        \exp(-f(2/\eta, d)) \leq \mu(B_{\eta/2}(X)) \leq \mu(\{Z \in \Omega ~|~ \langle Z - A_0, Y \rangle \geq 0\}).
    \]
    That is, $A_0$ is in the $(0,\exp(-f(\frac{2}{\eta}, d)))$-interior of $\mu$.
    
    Now for $A$ in the $(\eta,0)$-interior of $\mu$, we have that the $\frac{\eta}{2}$-ball about $A$ is contained in the $(\frac{\eta}{2},0)$-interior of $\mu$.
    Therefore $A$ is in the $(\frac{\eta}{2},\exp(-f(\frac{2}{\eta}, d)))$-interior of $\mu$. By Lemma \ref{lem:twoparam_bounding}, this implies $\|Y^\star\| \leq 2\eta^{-1} \cdot f(2\eta^{-1}, d)$.
\end{proof}

\begin{remark}
    Note that Theorem \ref{thm:bounding-main} is immediately applicable to uniform discrete measures on (singly) exponentially sized sets $S$.
    In particular, such a measure is automatically balanced with constant bound $f = \log|S|$.
\end{remark}

\subsection{Rank-$k$ projections}

We now prove bounding box result for $\mathcal{P}_k$, by showing that $\mu_k$ is balanced and applying the previous theorem.
Note that in this case $\mathcal{L}(X) = B$ reduces to $\Tr(X) = k$, and so $V_\mathcal{L}$ is the set of traceless Hermitian matrices in this case.

\begin{corollary}[Bounding box for $\mathcal{P}_k$] \label{cor:boundingbox_rank_k}
    Let $\mu_k$ be the uniform distribution on $\mathcal{P}_k$.
    Then given $A$ in the $(\eta,0)$-interior of $\mu_k$, the optimal traceless solution $Y^\star$ of the corresponding dual program is such that $\|Y^\star\| \leq \frac{2n^2}{\eta} \log\left(\frac{8n\sqrt{k}}{\eta}\right)$.
\end{corollary}
\begin{proof}
    We prove that $\mu_k$ is balanced and then apply the previous proposition.
    The number of balls of size $\delta$ required to cover the unit ball in $\R^{n^2}$ (with Euclidean/Frobenius norm) is at most $(2n/\delta)^{n^2}$.
    Since the set of projections of rank $k$ is contained in the sphere of radius $\sqrt{k}$, we have that it requires at most $(2n\sqrt{k}/\delta)^{n^2}$ $\delta$-balls to cover all such projections.
    With this, there exists some $\delta$-ball (call it $B_\delta$) in this cover which contains at least $(2n\sqrt{k}/\delta)^{-n^2}$ of the mass of $\mu_k$.
    Pick some $X \in \mathcal{P}_k \cap B_\delta$, and let $B_{2\delta}(X)$ be the ball of radius $2\delta$ which is centered at $X$.
    Thus, in fact $B_{2\delta}(X)$ contains at least $(2n\sqrt{k}/\delta)^{-n^2}$ of the mass of $\mu_k$.
    By unitary invariance of $\mu_k$, we have that the ball of radius $2\delta$ about any point of $\mathcal{P}_k$ contains at least $(2n\sqrt{k}/\delta)^{-n^2}$ of the mass of $\mu_k$.
    That is, $\mu_k$ is $\delta$-balanced with bound $f(\delta^{-1}, n) = n^2 \log(4n\sqrt{k} \cdot \delta^{-1})$ for all $\delta > 0$.
    Applying the previous proposition then gives the result.
\end{proof}

\begin{remark} \label{rem:bound_dependence}
    In the discrete measure case, the authors of \cite{StraszakV19} were able to improve the dependence on $\eta$ of the bounding box from $\eta^{-1}$ to $\log(\eta^{-1})$.
    This leads to a max-entropy approximation algorithm which does not depend on $\eta$.
    One may then naturally ask whether or not this is possible for the bounding box for $\mu_k$ discussed here.
    The answer turns out to be ``no'', and this can be seen by considering the optimal $Y^\star = \diag(y_1,y_2)$ in the case of $n=2$ and $k=1$.
    Specifically one can show that for $A = \diag(\eta, 1-\eta)$, the value of $|y_1-y_2|$ is of the order $\eta^{-1}$ as $\eta \to 0$.
    Since the relative entropy of the optimal distribution is unbounded as $\eta$ approaches 0, approximation of $Y^\star$ cannot help us to improve the dependence of $|y_1-y_2|$ on $\eta^{-1}$.
\end{remark}

\subsection{Convex bodies} \label{sec:boundingbox_convex}

We now prove bounding box result for convex bodies.
Instead of applying the previous theorem directly, we make some simpler computations which are in the same spirit.

\begin{corollary}[Bounding box for convex bodies] \label{cor:boundingbox_convex}
    Let $\mu$ be the uniform distribution on a $d$-dimensional convex body $\Omega$ contained in a ball of radius $R$. (Note that $\mathcal{K} = \hull(\Omega) = \Omega$ in this case.) Then given $\alpha$ in the $(\eta,0)$-interior of $\mu$, the optimal solution $y^\star \in V_\mathcal{L}$ of the corresponding dual program is such that $\|y^\star\| \leq \frac{2d}{\eta} \log(\frac{4R}{\eta})$.
\end{corollary}
\begin{proof}
    Note that $\alpha$ in the $(\eta, 0)$-interior of $\mu$ is automatically in the $\left(\frac{\eta}{2}, \left(\frac{\eta}{4R}\right)^d\right)$-interior of $\mu$, since:
    \[
        \mu(B_{\eta/4}) \geq \frac{\mathrm{vol}(B_{\eta/4})}{\mathrm{vol}(B_R)} = \left(\frac{\eta}{4R}\right)^d.
    \]
    By Lemma \ref{lem:twoparam_bounding}, this implies $\|y^\star\| \leq \frac{2d}{\eta} \log(\frac{4R}{\eta})$.
\end{proof}

\section{Counting oracle for $\mathcal{P}_k$} \label{sec:counting_oracle}

In this section, we prove existence of a strong counting/integration oracle for the objective function of the dual program $\dual_k$.
Recall the dual objective function:
\[
    F_A(Y) = \langle Y, A \rangle + \mathcal{E}_k(Y) = \langle Y, A \rangle + \log \int_{\mathcal{P}_k} e^{-\langle Y, X \rangle} d\mu_k(X).
\]
We want to be able to efficiently compute this function and its gradient.
In this case of rank-$k$ projections, we make the simplifying assumption that $Y$ and $A$ are both diagonal.
This simplification is actually without loss of generality, due to the Schur-Horn theorem (Corollary \ref{cor:schur_horn_opt}) and unitary invariance of $\mu_k$.
Further, it is enough to consider only $\mathcal{E}_k(Y)$ (which is independent of $A$) since $\langle Y, A \rangle$ is linear and hence easy to handle.
This leads to the main theorem of this section, stated originally as Theorem \ref{thm:counting}.

\begin{theorem}[Counting oracle for $\mathcal{P}_k$] \label{thm:counting_in_section}
    There is an algorithm that, given $n \in \N$, $k \in [n]$, an $n \times n$ real diagonal matrix $Y = \diag(y)$, and a $\delta > 0$, returns numbers $\bar{E}, \bar{G}$ such that 
    \begin{enumerate}
        \item $|\bar{E} - \mathcal{E}_k(Y)| \leq \delta$
        \item $|\bar{G} - \nabla \mathcal{E}_k(Y)| \leq \delta$, 
    \end{enumerate}  
    where $\mathcal{E}_k$ is the exponential integral defined above (and in Definition \ref{def:exp_integral}). The running time of the algorithm is polynomial in $n$, $\log(\frac{1}{\delta})$, and the number of bits needed to represent $e^{-y_i}$ for any $i \in [n]$.
\end{theorem}

\noindent
The main tool we use to prove this theorem is a collection of explicit formulas for computing $\mathcal{E}_k$ and its gradient.
We first discuss this in full detail for the case of $k=1$.
After that, we discuss how to generalize the arguments to the $k > 1$ case.

\subsection{Algorithm for $k=1$}

In this section, we construct the strong counting/integration oracle for rank-1 projections by giving formulas for the function $\mathcal{E}_1$ and its gradient (Propositions \ref{prop:evaluation_formula} and \ref{prop:gradient_formula}).
Specifically, for diagonal $Y = \diag(y)$ with distinct entries $\lambda_1 > \cdots > \lambda_k$ with multiplicities $m_1,\ldots,m_k$, we can compute the following where $M(y)$ and $M_p(y)$ are matrices defined below:
\[
    \mathcal{E}_1(Y) = \sum_{i=1}^{n-1} \log(i) + \log\det(M(-y)) - \sum_{i < j} m_im_j \log(\lambda_i - \lambda_j),
\]
\[
    (\nabla \mathcal{E}_1(Y))_l = -\sum_{i \neq p} \frac{m_i}{\lambda_p - \lambda_i} - \frac{\det(M_p(-y))}{\det(M(-y))}.
\]
The only potentially hard part of computing these expressions is computing the determinants of $M(-y)$ and $M_p(-y)$.
It is a standard fact that one can compute a determinant in time polynomial in the number of bits needed to represent the matrix, so we just need to demonstrate that the matrices have the necessary bit complexity.
Considering Definitions \ref{def:eval_matrix} (for $\gamma = 1$) and \ref{def:grad_matrix} below, we see that the matrix entries depend on computing $e^{-y_i}$, $n!$, and $y_i^n$.
All of these can be computed in time polynomial in $n$ and number of bits needed to represent $e^{-y_i}$, which is exactly what we need.
That said, all we have left now is to prove the two formulas stated above, and we do this in the following sections.

\subsubsection{Evaluating the dual integral}

We now prove the main evaluation formulas for integrals on the manifold $\mathcal{P}_1$.
Throughout we will often consider integrals on the unit sphere in $\C^n$, denoted $S_\C^n$, instead of on $\mathcal{P}_1$ directly, and we will let $\mu_{S_\C^n}$ refer to the Haar measure on the unit sphere.
Note that the transfer of formulas from the sphere to $\mathcal{P}_1$ is straightforward, as given by $(2)$ of Proposition \ref{prop:evaluation_formula}.
First, we define a parameterized matrix of a particular form which will show up many times in our computations.

\begin{definition}[Matrix for dual integral, $k=1$] \label{def:eval_matrix}
    Given $y_1,\ldots,y_n \in \R$, let $\lambda_1 < \cdots < \lambda_k$ denote the distinct values of $y_i$ with multiplicities $m_1,\ldots,m_k$. Given $\gamma$, we define an $n \times n$ matrix $M(y,\gamma)$ as follows:
    \[
    M(y,\gamma) := \left[
        \begin{matrix}
            1 & 0 & \cdots & 0 & 1 & 0 & \cdots \\
            \lambda_1 & 1 & \cdots & 0 & \lambda_2 & 1 & \cdots \\
            \lambda_1^2 & \binom{2}{1}\lambda_1 & \cdots & 0 & \lambda_2^2 & \binom{2}{1}\lambda_2 & \cdots \\
            \vdots & \vdots & \ddots & \vdots & \vdots & \vdots & \cdots \\
            \lambda_1^{n-2} & \binom{n-2}{1}\lambda_1^{n-3} & \cdots & \binom{n-2}{m_1-1}\lambda_1^{n-m_1-1} & \lambda_2^{n-2} & \binom{n-2}{1}\lambda_2^{n-3} & \cdots \\
            \frac{e^{\gamma\lambda_1}}{0!} & \frac{\gamma e^{\gamma\lambda_1}}{1!} & \cdots & \frac{\gamma^{m_1-1} e^{\gamma \lambda_1}}{(m_1-1)!} & \frac{e^{\gamma\lambda_2}}{0!} & \frac{\gamma e^{\gamma\lambda_2}}{1!} & \cdots \\
        \end{matrix}
    \right].
    \]
    We also define $M(y) := M(y,1)$. Note that only one row of $M(y,\gamma)$ depends on $\gamma$.
\end{definition}

\noindent
We now state a lemma which gives the most basic result about integrals on $\mathcal{P}_1$.
Specifically, we state a well-known result for integrals of polynomial-like functions.
This proof is very related to the unitarily invariant inner product on homogeneous polynomials, which has many names in the literature: Bombieri inner product, Fischer-Fock inner product, Segal-Bargmann inner product, etc.
The following lemma is standard, see e.g. Lemma 3.2 of \cite{pinasco2012}.

\begin{lemma}[Bombieri inner product formula] \label{lem:bombieri}
    For $\alpha \in \{0,1,2,\ldots\}^n$ such that $\sum_i \alpha_i = d$, we have:
    \[
        \int |v|^{2\alpha} d\mu_{S_\C^n}(v) = \int \prod_i |v_i|^{2\alpha_i} d\mu_{S_\C^n}(v) = \binom{d}{\alpha}^{-1} \binom{d+n-1}{n-1}^{-1} = \frac{\alpha_1! \cdots \alpha_n!(n-1)!}{(d+n-1)!}.
    \]
    Here, $\binom{d}{\alpha}$ is the multinomial coefficient, and $\binom{d+n-1}{n-1}$ is the binomial coefficient.
\end{lemma}

\noindent
The next lemma then shows the connection between the integrals we want to compute and the Laplace transform.
As an immediately corollary, we obtain equality of $(3)$ and $(4)$ in Proposition \ref{prop:evaluation_formula} below in the case of distinct values of $y_1,\ldots,y_n$.

\begin{lemma} \label{lem:laplace_integral}
    For $y_1,\ldots,y_n \in \R$ and $x_n := 1-x_1-\cdots-x_{n-1}$, we have the following where $*$ denotes the usual integral convolution:
    \[
        \int_0^1 \int_0^{1-x_1} \cdots \int_0^{1-x_1-\cdots-x_{n-2}} e^{\langle y, x \rangle} dx_{n-1} \cdots dx_1 = \left.(e^{y_1 t} * \cdots * e^{y_n t})\right|_{t=1}.
    \]
    If $y_1 < y_2 < \cdots < y_n$, then we further have:
    \[
        \int_0^1 \int_0^{1-x_1} \cdots \int_0^{1-x_1-\cdots-x_{n-2}} e^{\langle y, x \rangle} dx_{n-1} \cdots dx_1 = \sum_{i=1}^n \frac{e^{y_i}}{\prod_{j \neq i} (y_i-y_j)} = \frac{\det(M(y))}{\prod_{i < j} (y_j-y_i)}.
    \]
\end{lemma}
\begin{proof}
    We first compute:
    \[
    \begin{split}
        \int_0^1 \int_0^{1-x_1} \cdots &\int_0^{1-x_1-\cdots-x_{n-2}} e^{\langle y, x \rangle} dx_{n-1} \cdots dx_1 \\
            &= \int_0^1 e^{y_1x_1} \cdots \int_0^{1-x_1-\cdots-x_{n-2}} e^{y_{n-1}x_{n-1}} e^{y_n(1-x_1-\cdots-x_{n-1})} dx_{n-1} \cdots dx_1 \\
            &= \int_0^1 e^{y_1x_1} \cdots \int_0^{1-x_1-\cdots-x_{n-3}} e^{y_{n-2}x_{n-2}} \left.(e^{y_{n-1} t} * e^{y_n t})\right|_{t=1-x_1-\cdots-x_{n-2}} dx_{n-2} \cdots dx_1 \\
            &= \cdots = \int_0^1 e^{y_1x_1} \left.(e^{y_2 t} * \cdots * e^{y_n t})\right|_{t=1-x_1} dx_1 \\
            &= \left.(e^{y_1 t} * \cdots * e^{y_n t})\right|_{t=1}.
    \end{split}
    \]
    Using the Laplace transform, we have $\mathcal{L}[e^{y_i t}](s) = \frac{1}{s-y_i}$ which implies:
    \[
        e^{y_1 t} * \cdots * e^{y_n t} = \mathcal{L}^{-1}\left[\frac{1}{(s-y_1)(s-y_2) \cdots (s-y_n)}\right](t).
    \]
    Assuming $y_1 < y_2 < \cdots < y_n$, we can use Lagrange interpolation to compute:
    \[
        \mathcal{L}^{-1}\left[\frac{1}{(s-y_1)(s-y_2) \cdots (s-y_n)}\right](t) = \mathcal{L}^{-1}\left[\frac{1}{Q'(y_1) \cdot (s-y_1)} + \cdots + \frac{1}{Q'(y_n) \cdot (s-y_n)}\right](t).
    \]
    Here, $Q(s) := (s-y_1)(s-y_2) \cdots (s-y_n)$. With this we have:
    \[
        \mathcal{L}^{-1}\left[\frac{1}{Q'(y_1) \cdot (s-y_1)} + \cdots + \frac{1}{Q'(y_n) \cdot (s-y_n)}\right](t) = \frac{e^{y_1 t}}{Q'(y_1)} + \cdots + \frac{e^{y_n t}}{Q'(y_n)} = \sum_{i=1}^n \frac{e^{y_i t}}{\prod_{j \neq i} (y_i-y_j)}.
    \]
    Plugging in $t=1$ gives the first equality in the second statement. To see the last equality, notice that because $y_1 < \cdots < y_n$, the expression for $\det(M(y))$ will be a sum of exponentials multiplied by Vandermonde determinants (expand along the last row of $M(y)$). The result follows, taking care to keep track of signs.
\end{proof}

\noindent
We now state and prove the full evaluation formula for $\mathcal{P}_1$.
The two most involved parts of the proof are showing equality of $(1)$ and $(4)$ on polynomials and showing equality of $(4)$ and $(5)$ for non-distinct values of $y_1,..,y_n$.

\begin{proposition}[Evaluating the dual integral, $k=1$] \label{prop:evaluation_formula}
    Fix $n \in \N$, and let $\mu_{S_\C^n}, \mu_1, \mu_{\Delta_1}$ be the uniform probability distributions on the complex unit sphere in $\C^n$, on $\mathcal{P}_1$, and on the standard simplex in $\R^n$, respectively. For a given analytic function $f$ on the standard simplex the following expressions are equal:
    \begin{enumerate}
        \item $\displaystyle \int_{S_\C^n} f(|v_1|^2,\ldots,|v_n|^2) d\mu_{S_\C^n}(v)$,
        \item $\displaystyle \int_{\mathcal{P}_1} f(\diag(X)) d\mu_1(X)$,
        \item $\displaystyle \int_{\Delta_1} f(x) d\mu_{\Delta_1}(x)$,
        \item $\displaystyle (n-1)! \int_0^1 \int_0^{1-x_1} \cdots \int_0^{1-x_1-\cdots-x_{n-2}} f(x_1,\ldots,x_{n-1},1-x_1-\cdots-x_{n-1}) dx_{n-1} \cdots dx_1$.
    \end{enumerate}
    If $f(x) = e^{\langle y, x \rangle}$ for some real $y$ with distinct entries $\lambda_1 < \cdots < \lambda_k$ with multiplicities $m_1,\ldots,m_k$, then we have another equal expression:
    \begin{enumerate}
        \setcounter{enumi}{4}
        \item $\displaystyle (n-1)! \frac{\det(M(y))}{\prod_{i < j} (\lambda_j - \lambda_i)^{m_im_j}}$.
    \end{enumerate}
\end{proposition}
\begin{proof}
    First, for the equality of $(1)$ and $(2)$, note that $\mu_1$ is the pushforward measure of $\mu_{S_\C^n}$ through the map $\psi: S_\C^n \to \mathcal{P}_1$ given by $\psi: v \mapsto vv^*$. (To see this, note that $\psi$ is unitarily invariant and $\mu_{S_\C^n}$ and $\mu_1$ are the unique unitarily invariant measures on their respective domains.) With this, we then have:
    \[
        \int_{\mathcal{P}_1} f(\diag(X)) d\mu_1(X) = \int_{S_\C^n} f(\diag(\psi(v))) d_{S_\C^n}(v) = \int_{S_\C^n} f(|v_1|^2,\ldots,|v_n|^2) d_{S_\C^n}(v).
    \]
    That is, $(1)$ and $(2)$ are equal.

    Next, the equality of $(3)$ and $(4)$ follows from the fact that the map between the two domains of integration (both of which are simplices) is affine. Therefore the determinant of the Jacobian is a constant, and so we only need to integrate over a constant function to determine that constant. A simple induction shows that it is $(n-1)!$.
    
    To prove the equality of $(1)$ and $(4)$, we compute the integrals on a given monomial $x^m := x_1^{m_1} \cdots x_{n-1}^{m_{n-1}} (1-x_1-\cdots-x_{n-1})^{m_n}$. First, by Lemma \ref{lem:bombieri} we have:
    \[
        \int |v|^{2m} d\mu_{S_\C^n}(v) = \frac{m_1! \cdots m_n! (n-1)!}{(|m|+n-1)!} = \binom{|m|+n-1}{m_1,\ldots,m_n,n-1}^{-1}.
    \]
    Now, note:
    \[
        \int_0^1 x^l (1-x)^m dx = \sum_{k=0}^m \binom{m}{k} (-1)^k \int_0^1 x^{l+k} dx = \frac{l!m!}{(l+m+1)!} \sum_{k=0}^m \frac{\prod_{j \neq k} (l+j+1)}{(-1)^k k! (m-k)!} = \frac{l!m!}{(l+m+1)!}.
    \]
    The last equality is due to Lagrange interpolation, considering the sum as a function of $n$. We further have:
    \[
        \frac{l!m!}{(l+m+1)!} = \binom{l+m+2-1}{l,m,2-1}^{-1}.
    \]
    That is, we have equality whenever $n=2$, proving the base case. The rest of the proof goes by induction. First we compute for $\alpha = 1-x_1-\cdots-x_{n-2}$:
    \[
        \int_0^\alpha x_{n-1}^j (\alpha-x_{n-1})^k dx_{n-1} = \int_0^1 (\alpha u)^j (\alpha - \alpha u)^k \alpha du = \alpha^{j+k+1} \int_0^1 u^j (1-u)^k du = \frac{\alpha^{j+k+1} \cdot j!k!}{(j+k+1)!}.
    \]
    With this, we then compute the following by induction, letting $\beta = 1-x_1-\cdots-x_{n-3}$:
    \[
    \begin{split}
        (n-1)! \int_0^1 \cdots &\int_0^\alpha x_1^{m_1} \cdots x_{n-1}^{m_{n-1}} (\alpha-x_{n-1})^{m_n} dx_{n-1} \cdots dx_1 \\
            &= \frac{m_{n-1}!m_n!}{(m_{n-1}+m_n+1)!} \cdot (n-1)! \int_0^1 \cdots \int_0^{\alpha+x_{n-2}} x_1^{m_1} \cdots x_{n-2}^{m_{n-2}} \alpha^{m_{n-1}+m_n+1} dx_{n-2} \cdots dx_1 \\
            &= \frac{m_{n-1}!m_n!(n-1)}{(m_{n-1}+m_n+1)!} \left[(n-2)! \int_0^1 \cdots \int_0^\beta x_1^{m_1} \cdots x_{n-2}^{m_{n-2}} (\beta-x_{n-2})^{m_{n-1}+m_n+1} dx_{n-2} \cdots dx_1\right] \\
            &= \frac{m_{n-1}!m_n!(n-1)}{(m_{n-1}+m_n+1)!} \cdot \frac{m_1! \cdots m_{n-2}! (m_{n-1}+m_n+1)! (n-1-1)!}{(m_1+\cdots+m_n+1+n-1-1)!} \\
            &= \frac{m_1! \cdots m_n! (n-1)!}{(m_1+\cdots+m_n+n-1)!}.
    \end{split}
    \]
    This completes the proof of equality of $(1)$ and $(4)$.
    
    Finally, we prove the equality of $(4)$ and $(5)$ for $f(x) = e^{\langle y, x \rangle}$. Note that if $y_1 < \cdots < y_n$, then the result follows from the previous lemma. Otherwise, the expression in $(4)$ (for this function $f$) is continuous in $y_1,\ldots,y_n$, and so we can limit the expression for distinct eigenvalues. That said, we let $y'_1 < \cdots < y'_n$ be distinct values near to the $y_i$, and we apply L'Hoptial's rule to $\frac{\det(M(y'))}{\prod_{i < j} (y'_j - y'_i)}$ based on the multiplicities of the $y_i$. Specifically, for each $i \in [k]$ we apply the following differential operator to numerator and denominator (let $\partial_i := \partial_{y'_i}$):
    \[
        D_i := \prod_{j=1}^{m_i} \partial_{m_1+\cdots+m_{i-1}+j}^{j-1} = \partial_{m_1+\cdots+m_{i-1}+1}^0 \partial_{m_1+\cdots+m_{i-1}+2}^1 \cdots \partial_{m_1+\cdots+m_i}^{m_i-1}.
    \]
    The powers here correspond to the number of terms of the denominator of $\frac{\det(M(y'))}{\prod_{i < j} (y'_j - y'_i)}$ which will vanish when the $m_i$ values of $y'_1,\ldots,y'_n$ limit to $\lambda_i$. That said, we now want to compute:
    \[
        \left.\frac{D_1 \cdots D_k \det(M(y'))}{D_1 \cdots D_k \prod_{i < j} (y'_j - y'_i)}\right|_{y' = \lambda}.
    \]
    We first compute the denominator via the product rule, noting that the only nonzero term occurs whenever all derivatives from a given $D_i$ are applied to differences of eigenvalues corresponding to $\lambda_i$:
    \[
        \text{denominator} = \left(\prod_{i=1}^k \prod_{p=0}^{m_i-1} p!\right) \cdot \prod_{i < j} (\lambda_j - \lambda_i)^{m_im_j}.
    \]
    We next compute the numerator using the fact that exactly one row of the matrix depends on any given $y'_i$, as so we can apply the derivatives to the appropriate rows. Further, this means the numerator can still be expressed as a determinant. We also incorporate the factorials in the denominator expression above, by dividing each column by the appropriate factorial:
    \[
        \frac{\text{numerator}}{\prod_{i=1}^k \prod_{p=0}^{m_i-1} p!} = \left.\frac{D_1 \cdots D_k \det(M(y'))}{\prod_{i=1}^k \prod_{p=0}^{m_i-1} p!}\right|_{y'=\lambda} = \det(M(y)).
    \]
    This gives the result.
\end{proof}

\begin{remark} \label{rem:eval_order}
    Even though we assume the distinct values of $y_1,\ldots,y_n$ to be in increasing order in the previous result and in Definition \ref{def:eval_matrix}, we actually don't need this. Note that swapping the order of $\lambda_i$ and $\lambda_{i+1}$ affects the numerator and denominator of the expression $(4)$ in the same way, by multiplication by $(-1)^{m_im_{i+1}}$.
\end{remark}

\subsubsection{Computing the gradient}

We now compute the gradient of $\mathcal{E}_1(Y)$ for $Y = \diag(y)$ using the above formulas.
The first thing to note is that we can use an argument similar to what we used in the proof of the evaluation formula.
Specifically, note the following expression where $\partial_{y_l} := \frac{\partial}{\partial y_l}$:
\[
    \partial_{y_l} \mathcal{E}_1(Y) = \partial_{y_l} \log \int e^{-\langle Y, X \rangle} d\mu_1(X) = \frac{\int -X_{ll} \cdot e^{-\langle Y, X \rangle} d\mu_1(X)}{\int e^{-\langle Y, X \rangle} d\mu_1(X)}.
\]
In particular, we obtain the following bound where $y_1 \geq \cdots \geq y_n$ are the entries of diagonal $Y$:
\[
    |\partial_{y_l} \mathcal{E}_1(Y)| \leq \frac{\int e^{-\langle Y, X \rangle} d\mu_1(X)}{\int e^{-\langle y_l I_n, X \rangle} d\mu_1(X)} = e^{y_1} \int e^{-\langle Y, X \rangle} d\mu_1(X).
\]
From these observations, we have that $\partial_{y_l} \mathcal{E}_1(Y)$ is continuous on diagonal matrices $Y$.
Therefore, to compute the gradient we can first assume that $y_l$ is distinct from the other diagonal entries, and then limit via L'Hopital's rule (as in the proof of the evaluation formula).
We do exactly this to prove the gradient formula, after defining another parameterized matrix.

\begin{definition}[Matrix for gradient formula, $k=1$] \label{def:grad_matrix}
    Given $y_1,\ldots,y_n \in \R$, let $\lambda_1 < \cdots < \lambda_k$ denote the distinct values of $y_i$ with multiplicities $m_1,\ldots,m_k$. Given $p \in [k]$, we define an $n \times n$ matrix $M_p(y)$ as the matrix which differs from $M(y)$ in one column, given as follows:
    \[
    M_p(y) := \left[
        \begin{matrix}
            \cdots & 1 & 0 & \cdots & 0 & 0 & \cdots \\
            \cdots & \lambda_p & 1 & \cdots & 0 & 0 & \cdots \\
            \cdots & \lambda_p^2 & \binom{2}{1}\lambda_p & \cdots & 0 & 0 & \cdots \\
            \cdots & \vdots & \vdots & \ddots & \vdots & \vdots & \cdots \\
            \cdots & \lambda_p^{n-2} & \binom{n-2}{1}\lambda_p^{n-3} & \cdots & \binom{n-2}{m_p-2}\lambda_p^{n-m_p} & \binom{n-2}{m_p}\lambda_p^{n-m_p-2} & \cdots \\
            \cdots & \frac{e^{\lambda_p}}{0!} & \frac{e^{\lambda_p}}{1!} & \cdots & \frac{e^{\lambda_p}}{(m_p-2)!} & \frac{e^{\lambda_p}}{m_p!} & \cdots \\
        \end{matrix}
    \right].
    \]
    That is, $\frac{\partial_{\lambda_p}}{m_p}$ is applied to the right-most column of $M(y)$ that depends on $\lambda_p$.
\end{definition}

\begin{proposition}[Gradient formula, $k=1$] \label{prop:gradient_formula}
    Assume $y_1,\ldots,y_n$ are the diagonal values of diagonal $Y$, with distinct values $\lambda_1 > \cdots > \lambda_k$ and multiplicities $m_1, \ldots, m_k$. Letting $p$ be such that $y_l = \lambda_p$, we have the following expression:
    \[
        (\nabla \mathcal{E}_1(Y))_l = -\sum_{i \neq p} \frac{m_i}{\lambda_p - \lambda_i} - \frac{\det(M_p(-y))}{\det(M(-y))}.
    \]
\end{proposition}
\begin{proof}
    We first assume that $y_l$ is distinct from $\lambda_p$, and then we limit at the end.
    Specifically, we assume the distinct values of $y_1,\ldots,y_n$ are $\lambda_1 > \cdots > \lambda_p > y_l > \lambda_{p+1} > \cdots > \lambda_k$, where now the multiplicity of $\lambda_p$ is now one less than it was originally. We let $y'$ denote these new values of $y$ (with $y_l$ possibly changed) and let $m_i'$ denote these new multiplicities (only $m_p$ decreased by 1).
    We now want to compute:
    \[
    \begin{split}
        \left.\partial_{y_l} \mathcal{E}_1(Y)\right|_{y_l=\lambda_p} &= \left.\partial_{y_l} \log\left(\frac{(n-1)!}{\prod_{i < j} (\lambda_i - \lambda_j)^{m_i'm_j'}} \cdot \frac{\det(M(-y'))}{\prod_{i \leq p} (\lambda_i-y_l)^{m_i'} \prod_{i > p} (y_l-\lambda_i)^{m_i'}}\right)\right|_{y_l=\lambda_p} \\
            &= \left.\partial_{y_l} \left[\log\left(\frac{\det(M(-y'))}{(\lambda_p-y_l)^{m_p-1}}\right) - \sum_{i < p} m_i \log(\lambda_i-y_l) - \sum_{i > p} m_i \log(y_l-\lambda_i) \right]\right|_{y_l=\lambda_p}.
    \end{split}
    \]
    It is at this point that we limit $y_l$ to $\lambda_p$ and use the L'Hopital's rule argument.
    (Recall the above discussion which describes why this argument is valid.)
    We want to apply this argument to the following part of the above expression:
    \[
        \left.\partial_{y_l} \log\left(\frac{\det(M(-y'))}{(\lambda_p-y_l)^{m_p-1}}\right)\right|_{y_l=\lambda_p} = \left.\frac{(\lambda_p-y_l) \cdot \partial_{y_l} \det(M(-y')) + \det(M(-y')) \cdot (m_p-1)}{\det(M(-y')) \cdot (\lambda_p-y_l)}\right|_{y_l=\lambda_p}.
    \]
    The key is to notice that the denominator contains exactly $m_p'+1 = m_p$ factors of $(\lambda_p-y_l)$ up to scalar, where $m_p'$ factors come from the determinant.
    With this, we apply $\partial_{y_l}$ to the numerator and denominator $2m_p'$ times and then set $y_l = \lambda_p$.
    Computing this for the denominator is straightforward:
    \[
        \text{denominator} = \left.\partial_{y_l}^{m_p} \det(M(-y')) \cdot (\lambda_p-y_l)\right|_{y_l = \lambda_p} = (-1)^{m_p} m_p! \det(M(-y)).
    \]
    The computation is easy here for the same reason as in the proof of Proposition \ref{prop:evaluation_formula}: using the product rule for all the derivatives only leaves a single term which does not evaluate to zero once we set $y_l = \lambda_p$.
    A similar thing happens for the numerator, which yields:
    \[
    \begin{split}
        \text{numerator} &= \left.\partial_{y_l}^{m_p} \left[(\lambda_p-y_l) \cdot \partial_{y_l} \det(M(-y')) + \det(M(-y')) \cdot (m_p-1)\right]\right|_{y_l = \lambda_p} \\
            &= -\left.\partial_{y_l}^{m_p} \det(M(-y'))\right|_{y_l=\lambda_p} \\
            &= (-1)^{m_p+1} m_p! \det(M_p(-y)).
    \end{split}
    \]
    With this, we have the following expression:
    \[
        \left.\partial_{y_l} \mathcal{E}_1(Y)\right|_{y_l=\lambda_p} = -\sum_{i \neq p} \frac{m_i}{\lambda_p - \lambda_i} - \frac{\det(M_p(-y))}{\det(M(-y))}.
    \]
    This completes the proof.
\end{proof}

\subsection{Algorithm for $k > 1$}

We now discuss how to generalize the formulas and arguments from the rank-1 case to the rank-$k$ case.
The computations done here are very similar to those given above, and so we will be a bit less explicit in what follows.
And although the matrices involved are a bit more complicated (see Definitions \ref{def:eval_matrix_rank_k} and \ref{def:grad_matrix_rank_k}), we still achieve the required bit complexity bounds.
Specifically, each of the entries of these matrices require a polynomial number of computations of $m!$, $y_i^m$, and $e^{-y_i}$ for $m \leq n$, and so the determinants can still be computed as efficiently as is necessary for Theorem \ref{thm:counting_in_section}.

We now state the explicit integral formulas for $\mathcal{E}_k$ and $\nabla \mathcal{E}_k$ which generalize those of the $k=1$ case of the previous section.
Our main tool to prove these formulas is the Harish-Chandra-Itzykson-Zuber formula (\cite{HarishChandra1957}, \cite{IZ1980}), given as follows.

\begin{theorem}[HCIZ formula]\label{thm:HCIZ}
    For $n \times n$ Hermitian matrices $Y$ and $B$ with distinct eigenvalues $y_1 < \cdots < y_n$ and $\beta_1 < \cdots < \beta_n$ respectively, we have the following where $\mu$ is the Haar measure on the unitary group $U(n)$:
    \[
        \int_{U(n)} e^{\langle Y, UBU^* \rangle} d\mu(U) = \left(\prod_{p=1}^{n-1} p!\right) \frac{\det([e^{y_i\beta_j}]_{1 \leq i,j \leq n})}{\prod_{i < j} (y_j - y_i)(\beta_j - \beta_i)}.
    \]
\end{theorem}

\noindent
Using the L'Hoptial's rule argument used in the rank-1 case, we can limit $B$ to the rank-$k$ PSD projection $\diag(1,\ldots,1,0,\ldots,0)$ to obtain a formula for $\mathcal{E}_k(Y)$ for $Y$ with distinct eigenvalues.
Using again the same sort of argument, we can then limit $Y$ to any real diagonal matrix (eigenvalues not necessarily distinct).
First, we need to define a parameterized matrix as in the rank-1 case.

\begin{definition}[Matrix for dual integral, $k>1$] \label{def:eval_matrix_rank_k}
    Given $y_1,\ldots,y_n \in \R$, let $\lambda_1 < \cdots < \lambda_k$ denote the distinct values of $y_i$ with multiplicities $m_1,\ldots,m_k$. Let the polynomial $q_{i,j}(t)$ be defined as follows: 
    \[
        q_{i,j}(t) = e^{-t} \frac{\partial_t^j}{j!}(t^i e^t) = \sum_{l=0}^{\min(i,j)} \binom{i}{l} \frac{t^{i-l}}{(j-l)!}.
    \]
    We define an $n \times n$ matrix $M^{(k)}(y)$ as follows:
    \[
        M^{(k)}(y) := \left[
        \begin{matrix}
            1 & 0 & \cdots & 0 & 1 & \cdots \\
            \lambda_1 & 1 & \cdots & 0 & \lambda_2 & \cdots \\
            \lambda_1^2 & \binom{2}{1} \lambda_1 & \cdots & 0 & \lambda_2^2 & \cdots \\
            \vdots & \vdots & \ddots & \vdots & \vdots & \cdots \\
            \lambda_1^{n-k-1} & \binom{n-k-1}{1} \lambda_1^{n-k-2} & \cdots & \binom{n-k-1}{m_1-1}\lambda_1^{n-k-m_1} & \lambda_2^{n-k-1} & \cdots \\
            e^{\lambda_1} q_{0,0}(\lambda_1) & e^{\lambda_1} q_{0,1}(\lambda_1) & \cdots & e^{\lambda_1} q_{0,m_1-1}(\lambda_1) & e^{\lambda_2} q_{0,0}(\lambda_2) & \cdots \\
            e^{\lambda_1} q_{1,0}(\lambda_1) & e^{\lambda_1} q_{1,1}(\lambda_1) & \cdots & e^{\lambda_1} q_{1,m_1-1}(\lambda_1) & e^{\lambda_2} q_{1,0}(\lambda_2) & \cdots \\
            \vdots & \vdots & \ddots & \vdots & \vdots & \cdots \\
            e^{\lambda_1} q_{k-1,0}(\lambda_1) & e^{\lambda_1} q_{k-1,1}(\lambda_1) & \cdots & e^{\lambda_1} q_{k-1,m_1-1}(\lambda_1) & e^{\lambda_2} q_{k-1,0}(\lambda_2) & \cdots \\
        \end{matrix}
        \right].
    \]
    Also, any term of the form $\lambda_i^m$ for $m < 0$ in the above matrix should be replaced by 0.
\end{definition}

\noindent
The matrix defined above and the arguments of the previous section then allow us to write down an explicit formula for $\mathcal{E}_k(Y)$.

\begin{corollary} [Evaluating the dual integral, $k>1$] \label{cor:dual_integral_k}
    Let $Y$ be an $n \times n$ Hermitian matrix $Y$ with eigenvalues $y_1, \ldots, y_n$ and distinct eigenvalues $\lambda_1 > \cdots > \lambda_k$ with multiplicities $m_1, \ldots, m_k$. We have the following:
    \[
        \int_{\mathcal{P}_k} e^{-\langle Y, X \rangle} d\mu_k(X) = \frac{\prod_{p=1}^{n-1} p!}{\prod_{p=1}^{n-k-1} p! \cdot \prod_{p=1}^{k-1} p!} \cdot \frac{\det(M^{(k)}(-y))}{\prod_{i < j} (\lambda_i - \lambda_j)^{m_im_j}}.
    \]
    This leads to a formula for $\mathcal{E}_k(Y)$:
    \[
        \mathcal{E}_k(Y) = \log\left(\frac{\prod_{p=1}^{n-1} p!}{\prod_{p=1}^{n-k-1} p! \cdot \prod_{p=1}^{k-1} p!}\right) + \log\det(M^{(k)}(-y)) - \sum_{i < j} m_im_j \log(\lambda_i - \lambda_j).
    \]
\end{corollary}

\noindent
Notice that this reduces to $(5)$ of Proposition \ref{prop:evaluation_formula} whenever $k=1$.
As in the $k=1$ case, we use the Schur-Horn theorem and unitary invariance to restrict the inputs of $\mathcal{E}_k(Y)$ to real diagonal matrices (see Section \ref{sec:schur_horn}).
Therefore, we only need to compute the gradient on the diagonal entries of $Y$.
The arguments are essentially the same as those of the $k=1$ case, again via L'Hoptial's rule, and so we state the gradient formula for $\mathcal{E}_k$  as a corollary without proof.
First though, we need to define another parameterized matrix for the gradient formula, as in the $k=1$ case.

\begin{definition}[Matrix for gradient formula, $k>1$] \label{def:grad_matrix_rank_k}
    Given $y_1,\ldots,y_n \in \R$, let $\lambda_1 < \cdots < \lambda_k$ denote the distinct values of $y_i$ with multiplicities $m_1,\ldots,m_k$. Let $q_{i,j}(t)$ be defined as in Definition \ref{def:eval_matrix_rank_k}. Given $p \in [k]$, we define an $n \times n$ matrix $M^{(k)}_p(y)$ as the matrix which differs from $M^{(k)}(y)$ in one column, given as follows:
    \[
    M^{(k)}_p(y) := \left[
        \begin{matrix}
            \cdots & 1 & 0 & \cdots & 0 & 0 & \cdots \\
            \cdots & \lambda_p & 1 & \cdots & 0 & 0 & \cdots \\
            \cdots & \lambda_p^2 & \binom{2}{1}\lambda_p & \cdots & 0 & 0 & \cdots \\
            \cdots & \vdots & \vdots & \ddots & \vdots & \vdots & \cdots \\
            \cdots & \lambda_p^{n-k-1} & \binom{n-k-1}{1}\lambda_p^{n-k-2} & \cdots & \binom{n-k-1}{m_p-2}\lambda_p^{n-k-m_p+1} & \binom{n-k-1}{m_p}\lambda_p^{n-k-m_p-1} & \cdots \\
            \cdots & e^{\lambda_p} q_{0,0}(\lambda_p) & e^{\lambda_p} q_{0,1}(\lambda_p) & \cdots & e^{\lambda_p} q_{0,m_p-2}(\lambda_p) & e^{\lambda_p} q_{0,m_p}(\lambda_p) & \cdots \\
            \cdots & e^{\lambda_p} q_{1,0}(\lambda_p) & e^{\lambda_p} q_{1,1}(\lambda_p) & \cdots & e^{\lambda_p} q_{1,m_p-2}(\lambda_1) & e^{\lambda_p} q_{1,m_p}(\lambda_p) & \cdots \\
            \cdots & \vdots & \vdots & \ddots & \vdots & \vdots & \cdots \\
            \cdots & e^{\lambda_p} q_{k-1,0}(\lambda_p) & e^{\lambda_p} q_{k-1,1}(\lambda_p) & \cdots & e^{\lambda_p} q_{k-1,m_p-2}(\lambda_p) & e^{\lambda_p} q_{k-1,m_p}(\lambda_p) & \cdots \\
        \end{matrix}
    \right].
    \]
    That is, $\frac{\partial_{\lambda_p}}{m_p}$ is applied to the right-most column of $M^{(k)}(y)$ that depends on $\lambda_p$. As in Definition \ref{def:eval_matrix_rank_k}, any term of the form $\lambda_i^m$ for $m < 0$ in the above matrix should be replaced by 0.
\end{definition}

\begin{corollary}[Gradient formula, $k>1$] \label{cor:gradient_formula_k}
    Assume $y_1,\ldots,y_n$ are the diagonal values of diagonal $Y$, with distinct values $\lambda_1 > \cdots > \lambda_k$ and multiplicities $m_1, \ldots, m_k$. Letting $p$ be such that $y_l = \lambda_p$, we have the following expression:
    \[
        (\nabla \mathcal{E}_k(Y))_l = -\sum_{i \neq p} \frac{m_i}{\lambda_p - \lambda_i} - \frac{\det(M^{(k)}_p(-y))}{\det(M^{(k)}(-y))}.
    \]
\end{corollary}

\subsection{Sampling from $\mathcal{P}_1$}\label{sec:sampling}

Given some real diagonal matrix $Y$ as in the previous section, we want to be able to sample from the measure on $\mathcal{P}_1$ given by $e^{-\langle Y, X \rangle}d\mu_1(X)$.
It is not immediately obvious how to do this on $\mathcal{P}_1$ itself, so we instead transfer the measure to a simpler domain.

Specifically, we use the Proposition \ref{prop:evaluation_formula} to transfer the sampling problem to the simplex.
Once on the simplex, we can apply standard techniques via the coordinate-wise cumulative distribution function (CDF).
That said, we now state the sampling process for $\mathcal{P}_1$ and then use the rest of the section to fill in the details.

\begin{proposition}[Rank-one Sampling] \label{prop:rank_one_sampling}
    Let $Y = \diag(y)$ be a real diagonal $n \times n$ matrix. The following process produces samples from the measure $e^{-\langle Y, X \rangle} d\mu_1(X)$ on $\mathcal{P}_1$.
    \begin{enumerate}
        \item Sample $v$ from the measure $e^{-\langle y, v \rangle} d\mu_{\Delta_1}(v)$ on the simplex $\Delta_1$ by iteratively sampling $v_i$ conditioned on $v_1,\ldots,v_{i-1}$.
        \item Sample $z_1,\ldots,z_n$ independently uniformly from the complex unit circle.
        \item Construct $X := (z\sqrt{v})(z\sqrt{v})^* \in \mathcal{P}_1$ where $z\sqrt{v}$ is the column vector $(z_1\sqrt{v_1}, \ldots, z_n\sqrt{v_n})$.
    \end{enumerate}
    Note that step $(1)$ is nontrivial, but we discuss how to sample coordinate-wise from the simplex below.
\end{proposition}
\begin{proof}
    We give a proof sketch here, leaving the details to the remainder of this section. First, the reason we are able to reduce to sampling on the simplex is due to Proposition \ref{prop:disintegration}. Specifically, let $\Phi: \mathcal{P}_1 \to \Delta_1$ be given by $\Phi: X \mapsto \diag(X)$. Then for any $v \in \Delta_1$, we have the following where $\mathbb{T}$ is the complex unit circle:
    \[
        \Phi^{-1}(v) = \{X \in \mathcal{P}_1 ~:~ v = \diag(X)\} = \{X \in \mathcal{P}_1 ~:~ X = (z\sqrt{v})(z\sqrt{v})^* \text{ for } z \in \mathbb{T}^n\}.
    \]
    The fact that $e^{-\langle Y, X \rangle} d\mu_1(X)$ is invariant under the action of conjugating $X$ by $\diag(z)$ for $z \in \mathbb{T}^n$ then implies that we can uniformly sample $z$ from $\mathbb{T}^n$ via Proposition \ref{prop:disintegration}.

    Second, sampling from the measure $e^{-\langle Y, X \rangle} d\mu_1(X)$ is nontrivial, but doable by sampling each coordinate conditioned on the previous coordinates sampled. To do this we need to be able to efficiently compute the cumulative density function (CDF) for the conditioned measures, and we discuss how to do this below. Once we have this, we can sample each conditioned coordinate using standard techniques; see \cite{owen2013}, Section 4.5.
\end{proof}

\noindent
For the case of $\mu_k$ for $k > 1$, we leave the question of sampling from the associated maximum entropy distributions as an open problem.

\paragraph{Transferring to the simplex.}

To transfer sampling from $\mathcal{P}_1$ to sampling from the simplex, we need a way of applying pushforward to sampling.
The way to do this is via \emph{disintegration} (see \cite{chang1997}), which we discuss in the following result.

\begin{proposition}[Pushforward sampling] \label{prop:disintegration}
    Let $X,Y$ be separable complete metric spaces, and let $\mu,\nu$ be probability measures on $X,Y$ respectively. Let $\Phi: X \to Y$ be a map such that $\nu$ is the pushforward measure of $\mu$. Further, for any $y \in Y$, let $\mu_y$ denote the measure on the fiber $\Phi^{-1}(y)$ given by disintegration: i.e., such that $\int_X f(x) d\mu(x) = \int_Y \int_{\Phi^{-1}(y)} f(x) d\mu_y(x) d\nu(y)$ for all measurable $f$ (see \cite{chang1997}). Then the measure on $X$ generated by sampling $y$ from $(Y,\nu)$, followed by sampling $x$ from $(\Phi^{-1}(y), \mu_y)$, is equal to $\mu$.
\end{proposition}
\begin{proof}
    Let $\gamma$ denote the measure on $X$ generated by the described two-step sampling process. For any measurable set $A$ we have the following, where $P_1$ and $P_2$ denote the probabilities according to the first and second steps of the process respectively:
    \[
    \begin{split}
        \gamma(A) &= P_1(y \in \Phi(A)) \cdot P_2(x \in \Phi^{-1}(y) \cap A ~|~ y \in \Phi(A)) \\
            &= \left(\int_{\Phi(A)} d\nu(y)\right) \cdot \frac{\int_{\Phi(A)} \int_{\Phi^{-1}(y) \cap A} d\mu_y(x) d\nu(y)}{\int_{\Phi(A)} \int_X d\mu_y(x) d\nu(y)} \\
            &= \int_{\Phi(A)} \int_{\Phi^{-1}(y) \cap A} d\mu_y(x) d\nu(y).
    \end{split}
    \]
    The second equality is just by definition of conditional probability. We then further have:
    \[
    \begin{split}
        \int_{\Phi(A)} \int_{\Phi^{-1}(y) \cap A} d\mu_y(x) d\nu(y) &= \int_{\Phi(A)} \int_{\Phi^{-1}(y)} 1_A(x) d\mu_y(x) d\nu(y) \\
            &= \int_Y \int_{\Phi^{-1}(y)} 1_A(x) d\mu_y(x) d\nu(y) \\
            &= \mu(A).
    \end{split}
    \]
    That is, $\gamma(A) = \mu(A)$.
\end{proof}

\paragraph{Computing the conditioned CDF.}

We now compute the conditioned CDF for each coordinate of the measure on the simplex in Corollary \ref{cor:cdf_formula_rank_one}, after a necessary lemma.
Note that the formula below in Corollary \ref{cor:cdf_formula_rank_one} is not given in full explicit detail.
However, the formula is still a constant times a determinant of a matrix, and expressions are given for the entries of that matrix.
They are in fact rational functions of polynomials in $\beta$, $y_i$, $e^{y_i\beta}$, and factorials at most $n$ (see below).
Therefore, the whole determinant is computable in time polynomial in $n$ and the number of bits needed to represent $e^{y_i\beta}$.

\begin{lemma}
    Fix $y_1,\ldots,y_n \in \R$, and let $\lambda_1 < \cdots < \lambda_k$ be the distinct values of $y_i$ with multiplicities $m_i$. For valid $\gamma > 0$ and $x_n := 1-\gamma-x_1-\cdots-x_{n-1}$, we have the following:
    \[
        \int_0^{1-\gamma} \int_0^{1-\gamma-x_1} \cdots \int_0^{1-\gamma-x_1-\cdots-x_{n-2}} e^{\langle y, x \rangle} dx_{n-1} \cdots dx_1 = \frac{\det(M(y,1-\gamma))}{\prod_{i < j} (\lambda_j - \lambda_i)^{m_im_j}}.
    \]
    Moreover, only one of the rows of the matrix $M(y,1-\gamma)$ depends on $\gamma$.
\end{lemma}
\begin{proof}
    The proof follows from a simple substitution ($u_i = \frac{x_i}{1-\gamma}$), applying Proposition \ref{prop:evaluation_formula}, and then multiplying and dividing factors of $(1-\gamma)$ in the rows and columns of $M((1-\gamma)y)$ to obtain $M(y,1-\gamma)$.

    To see this, we apply the change of variables $u_i = \frac{x_i}{1-\gamma}$ and Proposition \ref{prop:evaluation_formula}:
    \[
    \begin{split}
        \text{left-hand side} &= \int_0^1 \int_0^{1-u_1} \cdots \int_0^{1-u_1-\cdots-u_{n-2}} e^{(1-\gamma)\langle y, u \rangle} (1-\gamma)^{n-1} du_{n-1} \cdots du_1 \\
            &= (1-\gamma)^{n-1} \frac{\det(M((1-\gamma)y))}{\prod_{i < j} (1-\gamma)^{m_im_j}(\lambda_j-\lambda_i)^{m_im_j}}.
    \end{split}
    \]
    Note now that we can do the following to $M((1-\gamma)y)$ to make it so that only one of its rows depends on $\gamma$ (recall the definition of $M(y)$ from Definition \ref{def:eval_matrix}). First, divide the $i$th row of the matrix by $(1-\gamma)^{i-1}$ up to $i=n-1$. Then, for any $\lambda_p$ multiply the $j$th column depending on $\lambda_p$ by $(1-\gamma)^{j-1}$. Only the last row of the matrix obtained will depend on $\gamma$, and in fact this matrix is precisely $M(y, 1-\gamma)$.

    The process described above is equivalent to pulling out of the determinant a factor of $(1-\gamma)$ with the following exponent:
    \[
        \text{exponent of factor} = \sum_{i=0}^{n-2} i - \sum_{p=1}^k \sum_{j=0}^{m_p-1} j = \binom{n-1}{2} - \sum_{p=1}^k \binom{m_p}{2}.
    \]
    With this have that
    \[
        \text{left-hand side} = (1-\gamma)^\xi \cdot \frac{\det(M(y, 1-\gamma))}{\prod_{i < j}(\lambda_j-\lambda_i)^{m_im_j}},
    \]
    where $\xi = (n-1) - \sum_{i<j} m_im_j + \binom{n-1}{2} - \sum_{p=1}^k \binom{m_p}{2}$. Note that
    \[
        \xi = \frac{n(n-1)}{2} - \frac{1}{2}\left(\sum_{p=1}^k m_p\right)^2 + \frac{1}{2}\sum_{p=1}^k m_p = 0,
    \]
    since $\sum_{p=1}^k m_p = n$. The result follows.
\end{proof}

\begin{corollary}[Conditioned CDF formula] \label{cor:cdf_formula_rank_one}
    Fix $y_1,\ldots,y_n \in \R$, and let $\lambda_1 < \cdots < \lambda_p$ be the distinct values of $y_{k+1},\ldots,y_n$ with multiplicities $m_i$. Further, fix $x_1=\alpha_1$, \ldots, $x_{k-1}=\alpha_{k-1}$ and let $\alpha := \sum_{i=1}^{k-1} \alpha_i$. Also, let $x_n := 1-\alpha-x_k-\cdots-x_{n-1}$. The CDF denoted $F_k(\beta)$ for the simplex distribution $e^{\langle y, x \rangle}$, conditioned on the given values of $x_1,\ldots,x_{k-1}$, is given as follows for $\beta \in [0,1-\alpha]$ and $y' := (y_{k+1}, \ldots, y_n)$:
    \[
    \begin{split}
        \frac{F_k(\beta)}{(n-1)!} &:= \int_0^\beta \int_0^{1-\alpha-x_k} \cdots \int_0^{1-\alpha-x_k-\cdots-x_{n-2}} e^{\langle y, x \rangle} dx_{n-1} \cdots dx_k \\
            &= \frac{e^{y_1\alpha_1+\cdots+y_{k-1}\alpha_{k-1}}}{\prod_{i < j} (\lambda_j-\lambda_i)^{m_im_j}} \int_0^\beta e^{y_k x_k} \det(M(y',1-\alpha-x_k)) dx_k.
    \end{split}
    \]
    Recall the definition of $M(y, \gamma)$ from Definition \ref{def:eval_matrix}. Since only the last row of $M(y',1-\alpha-x_k)$ depends on $x_k$, the above integral can be passed to that row and computed explicitly when $y_k \neq \lambda_l$:
    \[
        \int_0^\beta e^{y_k x_k} \cdot \frac{(1-\alpha-x_k)^j e^{(1-\alpha-x_k)\lambda_l}}{j!} dx_k = e^{(1-\alpha)\lambda_l}\left[\sum_{i=0}^j \frac{(1-\alpha-\beta)^i e^{(y_k-\lambda_l)\beta} - (1-\alpha)^i}{i!(y_k-\lambda_l)^{j-i+1}}\right].
    \]
    If $y_k = \lambda_l$, we have the simpler expression, $e^{(1-\alpha)\lambda_l} \left[\frac{(1-\alpha)^{j+1} - (1-\alpha-\beta)^{j+1}}{(j+1)!}\right]$.
\end{corollary}
\begin{proof}
    We have:
    \[
        \frac{F_k(\beta)}{(n-1)!} = e^{y_1\alpha_1+\cdots} \int_0^\beta e^{y_kx_k} \left[\int_0^{1-\alpha-x_k} \cdots \int_0^{1-\alpha-x_k-\cdots-x_{n-2}} e^{y_{k+1} x_{k+1}+\cdots} dx_{n-1} \cdots dx_{k+1}\right] dx_k.
    \]
    We compute the inner expression using the previous lemma and $\gamma = \alpha+x_k$:
    \[
        \int_0^{1-\alpha-x_k} \cdots \int_0^{1-\alpha-x_k-\cdots-x_{n-2}} e^{y_{k+1} x_{k+1}+\cdots} dx_{n-1} \cdots dx_{k+1} = \frac{\det(M(y',1-\alpha-x_k))}{\prod_{i < j} (\lambda_j - \lambda_i)^{m_im_j}}.
    \]
    This then implies:
    \[
        \frac{F_k(\beta)}{(n-1)!} = \frac{e^{y_1\alpha_1+\cdots+y_{k-1}\alpha_{k-1}}}{\prod_{i < j} (\lambda_j-\lambda_i)^{m_im_j}} \int_0^\beta e^{y_k x_k} \det(M(y',1-\alpha-x_k)) dx_k.
    \]
    Since only one row of $M(y',1-\alpha-x_k)$ depends on $x_k$, we can compute the above integral entrywise on that row by linearity (after multiplying that row by the $e^{y_k x_k}$ factor). We now compute the final expression of the result, removing subscripts to simplify notation. First we make the change of variables $t = 1-\alpha-x$:
    \[
        \int_0^\beta e^{yx} \frac{(1-\alpha-x)e^{(1-\alpha-x)\lambda}}{j!} dx = \frac{e^{y(1-\alpha)}}{j!} \int_{1-\alpha-\beta}^{1-\alpha} t^j e^{t(\lambda-y)} dt.
    \]
    If $\lambda = y$, then we simply obtain $e^{(1-\alpha)\lambda} \left[\frac{(1-\alpha)^{j+1} - (1-\alpha-\beta)^{j+1}}{(j+1)!}\right]$. Otherwise, we use integration by parts to obtain:
    \[
        \frac{e^{y(1-\alpha)}}{j!} \int_{1-\alpha-\beta}^{1-\alpha} t^j e^{t(\lambda-y)} dt = e^{(1-\alpha)\lambda}\left[\sum_{i=0}^j \frac{(1-\alpha-\beta)^i e^{(y-\lambda)\beta} - (1-\alpha)^i}{i!(y-\lambda)^{j-i+1}}\right].
    \]
\end{proof}

\section{Computing maximum entropy measures} \label{sec:computing_max_entropy}

In this section we describe the entire algorithm for computing the optimum $Y^\star$ for the dual program $\dual_\mu(A)$, given some $A \in \mathcal{K} = \hull(\Omega)$.
The algorithm is essentially an application of the ellipsoid algorithm, based on a bounding box and a strong counting/integration oracle.
We first discuss this algorithm in general, and then apply it to specific cases based on results from the previous sections.

Before moving on, we discuss how the linear equality constraints $\mathcal{L}(X) = B$ come into play here.
We want to restrict our search space to the vector space $V_\mathcal{L}$ defined as the set of all $X$ such that $\mathcal{L}(X) = 0$.
The main reason for this is, since the constraints given by $\mathcal{L}(X) = B$ pick out an affine space in which $\mathcal{K}$ is full dimensional, restricting the search space to $V_\mathcal{L}$ causes the optimum $Y^\star$ to be unique.
Further, the bounding box results above apply specifically to this particular $Y^\star$.

Since we are given $\mathcal{L}$ effectively and explicitly, we assume for the ellipsoid algorithm that we can project the gradient (given by the strong counting oracle) onto $V_\mathcal{L}$.
That said, we will from now on assume $V_\mathcal{L}$ to be the domain in which we are optimizing.

\subsection{The ellipsoid framework} \label{sect:ellipsoid}

Using the standard argument via H\"older's inequality, we have that the dual objective function is convex:
\[
    F_A(Y) := \langle Y, A \rangle + \mathcal{E}_\mu(Y) = \langle Y, A \rangle + \log\left( \int_\Omega e^{-\langle Y, X \rangle} d\mu(X) \right).
\]
With this, the main optimization tool we use to approximate the the dual optimum $Y^\star$ is the ellipsoid algorithm.
Recall the following from \cite{SinghV14} Theorem 2.13, which was essentially taken from \cite{BentalN12}.

\begin{theorem}[Ellipsoid algorithm]\label{thm:ellipsoid}
    Given any $\beta > 0$ and $R > 0$, there is an algorithm which, given a strong first-order oracle for $F_A$, returns a $Y^\circ \in V_\mathcal{L}$ such that:
    \[
        F_A(Y^\circ) \leq \inf_{Y \in V_\mathcal{L}, \|Y\|_\infty \leq R} F_A(Y) + \beta\left(\sup_{Y \in V_\mathcal{L}, \|Y\|_\infty \leq R} F_A(Y) - \inf_{Y \in V_\mathcal{L}, \|Y\|_\infty \leq R} F_A(Y)\right).
    \]
    The number of calls to the strong first-order oracle for $F_A$ is bounded by a polynomial in $d$, $\log R$, and $\log (1/\beta)$. Here, $d$ is the dimension of the ambient Hilbert space in which $\Omega$ lies.
\end{theorem}

\noindent
We now prove the main theorem (Theorem \ref{thm:main_algorithm}) regarding the existence of an algorithm for approximating the optimum to the dual objective.

\begin{theorem}[Main algorithm, general case] \label{thm:main_algorithm_in_section}
    Let $\mu$ be a balanced measure on a domain $\Omega \subseteq \R^d$ contained in a ball of radius $r$. There exists an algorithm that, given a maximal set of linearly independent equalities $\mathcal{L}(X) = B$, an $A$ in the $\eta$-interior of $\mathcal{K} = \hull(\Omega)$, an $\epsilon > 0$, and a strong counting/integration oracle for the exponential integral $\mathcal{E}_\mu(Y)$, returns $Y^\circ \in V_\mathcal{L}$ such that
    \[
        F_A(Y^\circ) \leq F_A(Y^\star) + \epsilon,
    \]
    where $F_A$ is the objective function for the dual program $\dual_\mu(A)$, and $Y^\star \in V_\mathcal{L}$ is the optimum of the dual program. The running time of the algorithm is polynomial in $d$, $\eta^{-1}$, $\log(\epsilon^{-1})$, $\log(r)$, and the number of bits needed to represent $A$, $\mathcal{L}$, and $B$.
\end{theorem}
\begin{proof}
    To apply the ellipsoid algorithm, we need to choose the two parameters, $\beta$ and $R$. Since $\mu$ is balanced with some polynomial bound $f$, we choose for $R$ the bounding box given for balanced measures in Theorem \ref{thm:boundingbox}:
    \[
        R := 2\eta^{-1} \cdot f(2\eta^{-1}, d).
    \]
    So, the set $\{Y \in V_\mathcal{L} ~:~ \|Y\| \leq R\} \subset \{Y \in V_\mathcal{L} ~:~ \|Y\|_\infty \leq R\}$ contains the optimal $Y^\star$ for the dual program. Next, we need to choose $\beta$. Note that for $\|Y\|_\infty \leq R$ we have:
    \[
        |F_A(Y)| \leq |\langle Y, A \rangle| + \left|\log \int e^{-\langle Y, X \rangle} d\mu(X)\right| \leq r\|Y\|_\infty + r\|Y\|_\infty \leq 2r\sqrt{d}\|Y\| \leq 2rR\sqrt{d}.
    \]
    Therefore, choosing $\beta := \frac{\epsilon}{4rR\sqrt{d}}$ implies:
    \[
        \beta = \frac{\epsilon}{4rR\sqrt{d}} \leq \frac{\epsilon}{\sup_{Y \in V_\mathcal{L}, \|Y\|_\infty \leq R} F_A(Y) - \inf_{Y \in V_\mathcal{L}, \|Y\|_\infty \leq R} F_A(Y)}.
    \]
    The ellipsoid algorithm then guarantees a $Y^\circ$ such that:
    \[
        F_A(Y^\circ) \leq \inf_{Y \in V_\mathcal{L}, \|Y\|_\infty \leq R} F_A(Y) + \epsilon = F_A(Y^\star) + \epsilon.
    \]
    The number of calls to the strong counting oracle is polynomial in $d$, $\log(R) = \log(2\eta^{-1} \cdot f(2\eta^{-1}))$ and $\log(1/\beta) = \log(4rR\sqrt{d}\epsilon^{-1})$. Given the bounding box, each oracle call (now including computing $\langle Y, A \rangle$) can be implemented in time polynomial in $d$, $\eta^{-1}$, and the number of bits needed to represent $A$. This completes the proof.
\end{proof}

\subsection{Rank-$k$ Projections}

Next we apply the above result to the case of $\Omega = \mathcal{P}_k$ and $\mu = \mu_k$, i.e., the case of rank-$k$ projections.
To do so we make a few tweaks to the proof of the theorem for the general algorithm given in the previous section.
In particular, even though our domain $\mathcal{P}_k$ lies in the space of Hermitian matrices, our strong counting oracle for $\mathcal{E}_k$ only applies to real diagonal matrices $Y$.
That said, we now state the theorem for rank-$k$ projections and discuss such issues in the proof.

\begin{corollary}[Main algorithm, $\mathcal{P}_k$ case] \label{cor:rank_k_algorithm_in_section}
    There exists an algorithm that, given $n \in \N$, $k \in [n]$, $A$ in the $\eta$-interior of $\mathcal{P}_k$, and any $\epsilon > 0$, returns Hermitian $Y^\circ$ such that
    \[
        F_A(Y^\circ) \leq F_A(Y^\star) + \epsilon,
    \]
    where $F_A$ is the objective function for the dual program $\dual_k(A)$, and $Y^\star$ is an optimum of the dual program. The running time of the algorithm is polynomial in $n$, $\eta^{-1}$, $\log(\epsilon^{-1})$, and the number of bits need to represent $A$.
\end{corollary}
\begin{proof}
    The result essentially follows from the general case, with a few details that need to be dealt with. First, the maximal linear equalities for $\mathcal{P}_k$ boils down to something very simple within the space of Hermitian matrices. It is simply given by $\Tr(X) = k$. Thus, our search space $V_\mathcal{L}$ then becomes the set of traceless Hermitian matrices.
    
    Next, by unitary invariance of $\mu_k$ we can assume $A$ is diagonal by unitary conjugation. Once we obtain an approximate optimum $Y^\circ$ for the diagonalized $A$, we can obtain an approximate optimum for the original $A$ via conjugation by this unitary. Next, by the Schur-Horn theorem (see \S\ref{sec:schur_horn} and the discussion at the start of \S\ref{sec:counting_oracle}) we can further assume that $Y^\star$ is diagonal. That is, we can assume $A$ is real diagonal and restrict the domain of $F_A(Y)$ to real diagonal matrices $Y$.

    Once we make this simplifying assumption, we have access to a strong counting/integration oracle for $\mathcal{E}_k(Y)$ by Theorem \ref{thm:counting}. The proof for the general case then goes through (using this strong counting oracle and the bounding box result for rank-$k$ projections), giving the desired result.
\end{proof}

\section{The Goemans-Williamson measure}

We discuss two main features of the pushforward through $v \mapsto vv^\top$ of the Goemans-Williamson measure which are relevant to this paper.
We abuse notation in this section by letting $\mu_{\mathrm{GW}}$ refer to the pushforward measure on $\mathcal{V}_1$.
First, we prove that this measure is a max-entropy measure with respect to $\mathcal{V}_1$.
Second, we demonstrate that this measure \emph{cannot} be interpreted as a max-entropy measure on $\mathcal{P}_1$.
This second point demonstrates the fundamental difference between mex entropy measures on $\mathcal{V}_1$ and $\mathcal{P}_1$.

\subsection{Goemans-Williamson measure on $\mathcal{V}_1$ maximizes entropy}

In this section, we demonstrate how the measure associated to the Goemans-Williamson SDP rounding scheme can be interpreted as a max-entropy measure.
We describe it formally as follows.

\begin{definition}[Goemans-Williamson rounding scheme] \label{def:GW_rounding}
    Given an $n \times n$ real symmetric positive definite matrix $A$, let $V$ be a real $n \times n$ matrix such that $VV^\top = A$. The Goemans-Williamson rounding scheme proceeds as follows:
    \begin{enumerate}
        \item Sample a random standard Gaussian vector $g$ from $\R^n$.
        \item Return the rank-1 PSD matrix $(Vg)(Vg)^\top$.
    \end{enumerate}
    The measure associated to this sampling process we refer to as the \emph{Goemans-Williamson measure} and denote it $\mu_{\mathrm{GW}}$. This measure is supported on the rank-1 real symmetric PSD matrices, which is the set of extreme points of the real symmetric PSD cone.
\end{definition}

\noindent
Now let $m$ be the Lebesgue measure on $\R^n$, and let $\mu$ be the measure on the real symmetric PSD cone which is the pushforward of $m$ through the map $\Phi: x \mapsto xx^\top$.
With this we can also give an explicit description of the Goemans-Williamson measure.

\begin{proposition}[Goemans-Williamson density function] \label{prop:GW_density}
    The Goemans-Williamson measure on the set of rank-1 real symmetric PSD matrices is given by
    \[
        d\mu_{\mathrm{GW}}(X) \propto e^{-\langle \frac{1}{2}A^{-1}, X \rangle} d\mu(X),
    \]
    where $\mu$ is the pushforward of Lebesgue measure through $x \mapsto xx^\top$.
\end{proposition}
\begin{proof}
    Let $A = VV^\top$ as in the definition of $\mu_{\mathrm{GW}}$. Since a standard Gaussian $g$ is distributed according to $e^{-\frac{1}{2}\|g\|^2} dm(g)$, we can apply the change of variables formula to determine how $x := Vg$ is distributed. We have:
    \[
        x \sim e^{-\frac{1}{2}\|V^{-1}x\|^2} \cdot \det(V^{-1}) dm(x) = e^{-\langle \frac{1}{2}A^{-1}, xx^\top \rangle} \cdot \sqrt{\det(A^{-1})} dm(x).
    \]
    Considering the pushforward of this measure through $x \mapsto xx^\top$ gives the desired result.
\end{proof}

\noindent
Note that strong duality then immediately implies $\mu_{\mathrm{GW}}$ is a max-entropy measure with respect to $\mu$, since its density function is of the correct form.
To demonstrate this more concretely, we prove this explicitly below via an explicit formula $\mathcal{E}_\mu(Y)$.
First, the following observation tells us that it is sufficient to restrict $\mathcal{E}_\mu(Y)$ to positive definite $Y$.
\begin{lemma}
    If $Y$ is not PD, then $\int_{\mathcal{V}_1} e^{-\langle Y, X \rangle} d\mu(X) = +\infty$.
\end{lemma}
\begin{proof}
    Since $X$ is PSD, we have that $Y \prec Z$ implies $-\langle Y, X \rangle \geq -\langle Z, X \rangle$. Hence, to prove the result, we only need to show it for singular PSD matrices $Y$. Further, unitary invariance means we can restrict to diagonal $Y$. So, assume $Y = \diag(0,y_2,\ldots,y_n)$ for $y_i \geq 0$. Now consider:
    \[
    \begin{split}
        \int_{\mathcal{V}_1} e^{-\langle Y, X \rangle} d\mu(X) &= \int_{\R^n} e^{-\langle Y, xx^\top \rangle} dx = \int_{\R^n} e^{-\sum_{i=2}^n y_i |x_i|^2} dx \\
            &= \int_{-\infty}^\infty \int_{\R^{n-1}} e^{-\sum_{i=2}^n y_i |x_i|^2} d(x_2,\ldots,x_n) dx_1 \\
            &= \int_{-\infty}^\infty C dx_1 = +\infty.
    \end{split}
    \]
    Note that the inner integrand above does not depend on $x_1$, and so the evaluation of the inner integral yields some positive (possibly infinite) constant $C$ as written above.
\end{proof}

\noindent
We now give an explicit formula for $\mathcal{E}_\mu(Y)$ on positive definite $Y$.

\begin{proposition}[Lebesgue evaluation formula] \label{prop:lebesgue_formula}
    We have the following explicit expression for $\mathcal{E}_\mu(Y)$ for $n \times n$ real symmetric positive definite $Y$:
    \[
        \mathcal{E}_\mu(Y) := \log \int_{\mathcal{V}_1} e^{-\langle Y, X \rangle} d\mu(X) = \frac{n}{2}\log(\pi) - \frac{1}{2} \log\det(Y).
    \]
\end{proposition}
\begin{proof}
    Since $\mu$ is the pushforward measure of $m$ through $x \mapsto xx^t$, we have:
    \[
        \log \int_{\mathcal{V}_1} e^{-\langle Y, X \rangle} d\mu(X) = \log \int_{\R^n} e^{-\langle Y, xx^\top \rangle} dm(x) = \log\left(\pi^{n/2} \det(Y)^{-1/2}\right).
    \]
    The second equality is computed via the density function of the multivariate Guassian.
\end{proof}

\noindent
This then leads to the main result of this section.

\begin{corollary}[Max-entropy, SDP rounding] \label{cor:sdp_rounding}
    Given an $n \times n$ real symmetric positive definite marginals matrix $A$, the Goemans-Williamson measure $\mu_{\mathrm{GW}}$ is the max-entropy measure with respect to $\mu$, the pushforward through $x \mapsto xx^\top$ of the Lebesgue measure on $\R^n$. That is, $\mu_{\mathrm{GW}}$ is the optimal measure for $\primal_\mu(A)$.
\end{corollary}
\begin{proof}
    Proposition \ref{prop:lebesgue_formula} gives the following explicit expression for $\mathcal{E}_\mu(Y)$ with $n \times n$ real symmetric positive definite input $Y$:
    \[
        \mathcal{E}_\mu(Y) := \log \int e^{-\langle Y, X \rangle} d\mu(X) = \frac{n}{2}\log(\pi) - \frac{1}{2} \log\det(Y).
    \]
    By a standard computation, we then have the following:
    \[
        \nabla \mathcal{E}_\mu(Y) = -\frac{1}{2} \nabla \log\det(Y) = -\frac{1}{2} Y^{-1}.
    \]
    This implies the following regarding the gradient of the dual program objective $\dual_\mu(A)$ for positive definite $A$:
    \[
        0 = \nabla F_A(Y) = \nabla(\langle Y, A \rangle + \mathcal{E}_\mu(Y)) = A - \frac{1}{2}Y^{-1} \iff Y = \frac{1}{2} A^{-1}.
    \]
    That is, $Y^\star = \frac{1}{2} A^{-1}$ is the optimum for the dual program. By strong duality/Slater condition for $\mu$ (see Proposition \ref{prop:slaters_condition_GW}) and the density function for $\mu_{\mathrm{GW}}$ given in Proposition \ref{prop:GW_density} above, this implies the result.
\end{proof}

\subsection{Goemans-Williamson measure projected to the unit sphere does not maximize entropy} \label{sec:GWsphere}

In this section we show that the Hermitian version of the measure $\mu_{\mathrm{GW}}$ is not a max-entropy measure for $\mathcal{P}_1$. We do not recompute the density function for $\mu_{\mathrm{GW}}$ in the Hermitian case, but only say that Proposition \ref{prop:GW_density} can be adapted to show that in this case it is of the same form: $\nu(X) \propto e^{-\langle A_0, X \rangle}$ for some positive definite $A_0$.

We want to ``project'' the (Hermitian) SDP rounding measure onto $\mathcal{P}_1$, and we want to compute the density with respect to $\mu_1$. To do this, we first project the Lebesgue measure onto the complex unit sphere $S_\C^n$ and then pushforward through $x \mapsto xx^*$. We first state a few standard lemmas.

\begin{lemma}
    Let $f(z)$ be a Lebesgue measurable function on $\C^n$. Then:
    \[
        \int_{\C^n} f(z) dm(z) = \frac{2\pi^n}{(n-1)!} \int_{S_\C^n} \int_0^\infty f(rv) r^{2n-1} dr d\mu_{S_\C^n}(v).
    \]
\end{lemma}
\begin{proof}
    This is precisely the polar coordinates formula for Lebesgue measure in $\C^n \cong \R^{2n}$. The constant $\frac{2\pi^n}{(n-1)!}$ is the volume of the complex unit ball in $\C^n$.
\end{proof}

\noindent
This shows that the projected density $g$ can be computed from the Lebesgue density $f$ as follows:
\[
    g(v) = \frac{2\pi^n}{(n-1)!} \int_0^\infty f(rv) r^{2n-1} dr.
\]
We will now use the following lemma, which is standard.

\begin{lemma}
    For $n \in \N$ and $a > 0$, we have:
    \[
        \int_0^\infty r^{2n-1} e^{-ar^2} dr = \frac{(n-1)!}{2 a^n}.
    \]
\end{lemma}

\noindent
With this, we compute the following for $f(x) \sim e^{-\langle A, xx^* \rangle}$:
\[
\begin{split}
    g(v) &= \frac{2\pi^n}{(n-1)!} \int_0^\infty f(rv) r^{2n-1} dr \\
        &\propto \frac{2\pi^n}{(n-1)!} \int_0^\infty r^{2n-1} e^{-r^2\langle A, vv^* \rangle} dr \\
        &= \frac{\pi^n}{\langle A, vv^* \rangle ^n} \propto \langle A, vv^* \rangle^{-n}.
\end{split}
\]
That is, the projected density is proportional to $\langle A, vv^* \rangle^{-n}$ on the unit sphere. With this, we have the following interesting fact.

\begin{proposition}
    The ``projection'' of the (Hermitian) SDP rounding measure to $\mathcal{P}_1$ is not a max-entropy measure with respect to $\mu_1$ on $\mathcal{P}_1$.
\end{proposition}
\begin{proof}
    By strong duality, max-entropy densities in both contexts take the form $g(X) \propto e^{-\langle A, X \rangle}$. So, we just need to show that for all PD $B$ we have:
    \[
        \langle A, X \rangle^{-n} \not\sim e^{-\langle B, X \rangle}.
    \]
    This is straightforward, e.g. using the fact that the left-hand side is a rational function in $\Re(v_i),\Im(v_i)$ but the right-hand side is not.
\end{proof}

\section{Generalization of the maximum entropy framework to Lie groups}\label{sec:lie}

Recent work (e.g., \cite{christandl2014,BurgisserFGOWW18}) has demonstrated interesting connections between Lie theory and TCS, and the max-entropy framework
fits into this context as well.
In what follows we will briefly discuss the case of $\Omega = \mathcal{P}_k$ and $\mu = \mu_k$, as well as how this can be generalized.
However, a more detailed investigation of the computational aspects of the max-entropy framework in this context
is outside the scope of this paper.

We first describe the case of $\Omega = \mathcal{P}_k$ and $\mu = \mu_k$ in a more general way.
The unitary group $U(n)$ acts on the real vector space of $n \times n$ Hermitian matrices by conjugation.
This group action partitions the vector space into orbits, with $X$ and $Y$ being in the same orbit if and only if they have the eigenvalues.
Given any Hermitian matrix $F$, we denote the orbit corresponding to $F$ by $\mathcal{O}(F)$.

Consider now the matrix $P_k := \diag(1,\ldots,1,0,\ldots,0)$ where $k$ denotes the number of $1$s that appear in the matrix.
Then the orbit $\mathcal{O}(P_k)$ is precisely the set of rank-$k$ projections.
That is, $\mathcal{O}(P_k) = \mathcal{P}_k$, and so the unitarily invariant measure $\mu_k$ on $\mathcal{P}_k$ induces such a measure on $\mathcal{O}(P_k)$.
In fact such a unitarily invariant measure $\mu_F$ exists for any orbit $\mathcal{O}(F)$ allowing us to extend our maximum entropy framework to such orbits of $U(n)$.

This can be generalized beyond the group $U(n)$, to the general setting of a Lie group $G$ and its corresponding Lie algebra $\mathfrak{g}$ upon which $G$ naturally acts.
The primal and dual programs for this generalized setting are the same as in the general case, with one exception.
The element $F \in \mathfrak{g}$ is now an input, and any algorithm for approximating an optimum for $\dual_{\mu_F}(A)$ will necessarily depend on the complexity of $F$.
That said, strong duality holds in this case whenever $A$ is in the interior of $\mathcal{K} = \hull(\mathcal{O}(F)) \subset \mathfrak{g}$, and so the bounding box and the strong counting oracle are the two main results needed to obtain the polynomial-time ellipsoid-based algorithm described in this paper.
As an aside, in this case $\mathcal{K} = \hull(\mathcal{O}(F))$ is called an \emph{orbitope} (e.g., see \cite{sanyal2011,barvinok2008}).

Thus, the following optimization problem is a natural generalization of the (dual) maximum entropy problem considered in this paper.
The $G$-invariant inner product used in the exponent here can be derived from the so-called Killing form of $\mathfrak{g}$ when $G$ is compact (e.g., see \cite{knapp2013}, Corollary 4.26).
\[
    \inf_{Y \in \mathfrak{g}} F_A(Y) = \inf_{Y \in \mathfrak{g}} \left[\langle Y, A \rangle + \log \int_{\mathcal{O}(F)} e^{-\langle Y, X \rangle} d\mu_F(X)\right]
\]
Computability of this problem will be a subject of future work.

\section*{Acknowledgments} The authors would like to thank Sushant Sachdeva, Sebastien Bubeck, and Umesh Vazirani for useful discussions. 
They would also like to thank Simons Institute for the  Theory of Computing where this work was initiated. 
This research was partially supported by NSF CCF-1908347 grant and by Vetenskapsr\r adet.

\newpage

\bibliographystyle{plain}
\bibliography{references} 

\newpage

\appendix

\section{The dual formulation and strong duality}\label{sec:duality}

\subsection{The dual formulation} \label{sect:dualform}

The dual formulation $\dual_\mu(A)$ is given as follows, for $A \in \mathcal{K} = \hull(\Omega)$ and $Y$ in the ambient real inner product space $\R^d$:
\[
    \inf_Y F_A(Y) := \inf_Y \left[\langle Y, A \rangle + \log \int_\Omega e^{-\langle Y, X \rangle} d\mu(X)\right].
\]
In $\dual_\mu(A)$ we also assume a linear constraint on $Y$: $\mathcal{L}(Y) = 0$ where $X \in \Omega$ is such that $\mathcal{L}(X) = B$. We ignore this constraint for now, and deal with it in Lemma \ref{lem:linear_constraints} below.

To prove the form of the dual formulation given above, we write:
\[
    L(\nu, Y, z) = -\int \nu(X) \log\left(\nu(X)\right) d\mu(X) + \langle Y, A \rangle - \int \langle Y, X \rangle \nu(X) d\mu(X) + z - z \int \nu(X) d\mu(X).
\]
We now want to compute derivatives to connect this with the dual program. For any $f \in L^2(\mu)$, we compute:
\[
\begin{split}
    0 &= \left.\partial_t\right|_{t=0} L(\nu + t f, Y, z) \\
        &= -\int f(X) \log\left(\nu(X)\right) d\mu(X) - \int f(X) d\mu(X) - \int \langle Y, X \rangle f(X) d\mu(X) - z\int f(X) d\mu(X) \\
        &= -\int f(X) \left[\log\left(\nu(X)\right) + 1 + \langle Y, X \rangle + z\right] d\mu(X) \\
        &= -\left\langle f, \left[\log\left(\nu(X)\right) + 1 + \langle Y, X \rangle + z\right] \right\rangle.
\end{split}
\]
This immediately implies (almost everywhere, and we will suppress this caveat from now on):
\[
    \log\left(\nu(X)\right) + 1 + \langle Y, X \rangle + z = 0.
\]
This then gives
\[
    \nu(X) = \exp(-1 - z - \langle Y, X \rangle).
\]
and therefore:
\[
    -\int \left[\log\left(\nu(X)\right) + \langle Y, X \rangle + z\right] \nu(X) d\mu(X) = \int \nu(X) d\mu(X) = \int \exp(-1 - z - \langle Y, X \rangle) d\mu(X).
\]
Combining these observations:
\[
\begin{split}
    L(\nu, Y, z) &= \int \exp(-1 - z - \langle Y, X \rangle) d\mu(X) + \langle Y, A \rangle + z \\
        &= z + \langle Y, A \rangle + e^{-1-z} \int e^{-\langle Y, X \rangle} d\mu(X).
\end{split}
\]
Now, we compute:
\[
\begin{split}
    0 &= \partial_z L(\nu, Y, z) = 1 - e^{-1-z} \int e^{-\langle Y, X \rangle} d\mu(X) \\
        &\implies \int e^{-\langle Y, X \rangle} d\mu(X) = e^{1+z} \\
        &\implies z = \log\left(\int e^{-\langle Y, X \rangle} d\mu(X) \right) - 1.
\end{split}
\]
And finally
\[
    \inf_Y L(\nu, Y, z) = \inf_Y \left[ \langle Y, A \rangle + \log\left( \int e^{-\langle Y, X \rangle} d\mu(X) \right)\right],
\]
where $Y$ ranges over $\R^d$.

\begin{lemma} \label{lem:linear_constraints}
    Let $\mathcal{L}(X) = B$ be a set of linear constraints satisfied by all $X \in \Omega$. There exists an optimal solution $Y^\star$ to the dual program $\dual_\mu(A)$ if and only if there exists a solution $Z^\star$ to $\dual_\mu(A)$ restricted to $\mathcal{L}(Z^\star) = 0$.
\end{lemma}
\begin{proof}
    For any $Y$, consider the decomposition $Y = Z + Z^\perp$ where $\mathcal{L}(Z) = 0$ and $\langle Z^\perp, Y' \rangle$ for all $Y'$ such that $\mathcal{L}(Y') = 0$. Note that $A \in \Omega$ implies $\mathcal{L}(X-A) = 0$ for all $X \in \Omega$, and so $\langle Z^\perp, X-A \rangle = 0$. Letting $F_A(Y)$ denote the dual objective, this implies:
    \[
        F_A(Y) = \log \int e^{-\langle Y, X-A \rangle} d\mu(X) = \log \int e^{-\langle Z, X-A \rangle} d\mu(X) = F_A(Z).
    \]
    This completes the proof.
\end{proof}

\subsection{Strong duality under Slater's condition}\label{sect:sd_proof}

We now prove a general result above obtaining strong duality from a Slater-type condition. In the next section, we show that this Slater-type condition holds for the max-entropy program in general. We also give more concrete proofs for $\mu_k$ on $\mathcal{P}_k$ and $\mu$ on $\mathcal{V}_1$ in the following section.

\begin{proposition}[\bf Strong duality under Slater's condition]
    Let $V$ be a real inner product space such that $\Omega \subseteq V$. If for any $A$ in the relative interior of $\mathcal{K} = \hull(\Omega)$ there exists $\nu_A$ in the relative interior of the constraints of $\primal_\mu(A)$, then we have strong duality for any $A$ in the relative interior of $\mathcal{K}$.
\end{proposition}
\begin{proof}
    Without loss of generality, we may assume that $\Omega$ and $A$ have been translated such that $\mathcal{K} = \hull(\Omega)$ is full dimensional within a subspace of $V$, and that $A$ also lies in this subspace. We will let $W$ refer to this subspace, and now $\mathcal{K}$ has nonempty interior within $W$.

    We now follow some standard proofs of strong duality. Fix $A$ in the interior of $\mathcal{K}$, and let $\mathcal{D} := \{\nu: \supp(\mu) \to \R_+\}$. Then we only have linear equality constraints, which we denote collectively by:
    \[
        E(\nu) := \left(\int X\nu(X) d\mu(X) - A, \int \nu(X) d\mu(X) - 1\right) \in W \times \R.
    \]
    We also negate the objective function, denoting this negated function by $F$, and $p^\star$ denote its optimal value. Note that $p^\star$ is finite by the assumptions of the theorem. Now define:
    \[
        \mathcal{W} := \{(B,y,t) ~:~ F(\nu) \leq t \text{ and } E(\nu) = (B,y) \text{ for some } \nu \in \mathcal{D}\} \subseteq W \times \R \times \R.
    \]
    Note that this set is convex. Further, $(0,0,p^\star)$ is either on the boundary of $\mathcal{W}$ or outside of $\mathcal{W}$ by optimality. Hence, we can find a separating hyperplane, and therefore there exists some $0 \neq (C_0,z_0,s_0) \in W \times \R \times \R$ such that:
    \[
        \langle (C_0,z_0,s_0), (B,y,t) \rangle \geq s \cdot p^\star \qquad \text{for all } (B,y,t) \in \mathcal{V}.
    \]
    Note that $s_0 \geq 0$, or else we can pick $t$ very large to get a contradiction. We now demonstrate that in fact $s_0 > 0$. To get a contradiction, we suppose $s_0=0$ which gives:
    \[
        \langle (C_0,z_0), (B,y) \rangle \geq 0 \qquad \text{for all } (B,y,t) \in \mathcal{W}.
    \]
    For any $B = A + Z$ where $Z \in B_\epsilon(0) \subset W$, we have an interior solution $\nu_B$ by assumption. Therefore for all such $Z$:
    \[
        \langle C_0, Z \rangle = \langle C_0, Z \rangle + z_0 \cdot 0 = \langle (C_0,z_0), (B,y) \rangle \geq 0.
    \]
    This is only possible if $C_0 = 0$. By scaling $\nu_A$, we also see that $z_0 = 0$. This contradicts $(C_0,z_0,s_0) \neq 0$, and therefore $s_0 > 0$.
    
    To finish the proof, we define $C_0' := \frac{C_0}{s_0}$ and $z_0' := \frac{z_0}{s_0}$. We can then write:
    \[
        \langle (C_0',z_0',1), (B,y,t) \rangle \geq p^\star \qquad \text{for all } (B,y,t) \in \mathcal{W}.
    \]
    Note that, for the dual objective function $G(C,z) := \inf_{\nu \in \mathcal{D}} \left[F(\nu) + \langle (C,z), E(\nu) \rangle\right]$, we have:
    \[
        d^\star := \sup_{C,z} G(C,z) \geq G(C_0',z_0') = \inf_{(B,y,t) \in \mathcal{W}} \langle (C_0',z_0',1), (B,y,t) \rangle \geq p^\star.
    \]
    On the other hand, we have for all $C,z$ and $\nu_A$ satisfying constraints such that $F(\nu_A)$ is near optimal:
    \[
        G(C,z) = \inf_{\nu \in \mathcal{D}} [F(\nu) + \langle (C,z), E(\nu) \rangle] \leq F(\nu_A) + \langle (C,z), E(\nu_A) \rangle = p^\star + \epsilon + 0.
    \]
    Applying $\sup$ to $G(C,z)$ and letting $\epsilon \to 0$ implies $d^\star \leq p^\star$.
\end{proof}

\subsection{Slater's condition holds for general $\Omega$ and $\mu$} \label{sec:general_sd}

The dual formulation given in Section \ref{sect:dualform} implies a succinct representation of the optimal density $\nu$, given that we have strong duality. Here we prove strong duality for general $\Omega$ and $\mu$ by proving Slater's condition. The main thing needed for this is existence of an optimal $Y^\star$ for $\dual_\mu(A)$, given interiority of $A$. We prove this now.

\begin{lemma}[Existence of dual optimum] \label{lem:dual_optimum}
    If $A$ is in the interior of $\mathcal{K} = \hull(\Omega)$, then there exists $Y^\star$ which optimizes the dual program $\dual_\mu(A)$.
\end{lemma}
\begin{proof}
    By the previous lemma, we may assume that $\mathcal{K}$ is of full dimension in its ambient inner product space $V$. That said, note now that there is no closed half-space $H \subset V$ such that $A \in H$ and $\mu(H \cap \Omega) = 0$. Otherwise this would imply that $A$ is not in the interior of $\mathcal{K}$, since the interior of $H$ would not intersect $\Omega$, the support of $\mu$. Now suppose $A$ is in the $\eta$-interior of $\mathcal{K}$. Hence, for any $B \in B_{\eta/2}(A)$ and any half-space $H$ with $B$ on the boundary, there is an $\epsilon_{B,H} > 0$ such that $\mu(H \cap \Omega) = \epsilon_{B,H}$.
    
    We now prove that $\epsilon := \inf \epsilon_{B,H} > 0$. If not, then there is some sequence $(B_i,H_i)$ for which $\epsilon_{B_i,H_i} \to 0$. By identifying half-spaces about a point with the unit sphere, we have that the set of all possible $(B,H)$ pairs is compact. Thus we can assume $(B_i,H_i)$ is convergent, with limit $(B_0,H_0)$. Since every point $X$ of the interior of $H_0$ is eventually in the interior of $H_i$, we have that the measure in a small ball around any such $X$ is 0. Therefore the interior of $H_0$ does not intersect $\Omega$, the support of $\mu$. This implies $B_0$ is not in the interior of $\mathcal{K}$, a contradiction.
    
    Therefore, $\epsilon := \inf \epsilon_{B,H} > 0$. This in fact implies that $A$ is in the $(\frac{\eta}{2}, \epsilon)$-interior of $\mu$ (see Definition \ref{def:twoparam_interior}).

    Using the arguments of Lemma \ref{lem:twoparam_bounding}, for any $Y \in V$ we have:
    \[
    \begin{split}
        \epsilon &\leq \mu(\{X \in \Omega ~|~ \langle -Y, X-(A - (\eta/2) \cdot Y/\|Y\|)\rangle \geq 0\}) \\
            &= \mu(\{X \in \Omega ~|~ \langle -Y, X-A \rangle \geq (\eta/2) \cdot \|Y\|\}).
    \end{split}
    \]
    This implies:
    \[
        F_A(Y) = \log \int e^{\langle -Y, X-A \rangle} d\mu(X) \geq \log\left(\epsilon \cdot e^{(\eta/2) \cdot \|Y\|}\right) = \frac{\eta}{2} \|Y\| + \log(\epsilon).
    \]
    Hence, $\|Y\| > R$ implies a lower bound on the dual objective $F_A(Y)$, which goes to infinity as $R \to \infty$. Therefore $F_A$ must be minimized at some bounded point $Y^\star \in V$.
\end{proof}

\noindent
This lemma then implies Slater's condition in general.

\begin{theorem}[Strong duality] \label{thm:strongduality_app}
    Fix any $\mu$ with support $\Omega$ in a real Hilbert space $V$. If $A$ is in the interior of $\mathcal{K} = \hull(\Omega)$, then strong duality holds for $A$. In particular, the optimum density $\nu^\star$ is of the form:
    \[
        \nu^\star(X) \propto e^{-\langle Y^\star, X \rangle}
    \]
\end{theorem}
\begin{proof}
    Let $Y^\star$ be an optimum for $\dual_\mu(A)$, by Lemma \ref{lem:dual_optimum}. Then:
    \[
        0 = \nabla F_A(Y^\star) = \nabla \left[\langle Y, A \rangle + \log \int e^{-\langle Y, X \rangle} d\mu(X)\right]_{Y=Y^\star} = A - \frac{\int Xe^{-\langle Y^\star, X \rangle} d\mu(X)}{\int e^{-\langle Y^\star, X \rangle} d\mu(X)}.
    \]
    This precisely says that $A$ is the marginals matrix of the measure $\nu^\star(X) \propto e^{-\langle Y^\star, X \rangle}$. Therefore strong duality holds, since $\nu^\star$ is in the relative interior of the constraints of $\primal_\mu(A)$ for any $A$ in the interior of $\hull(\Omega)$.
\end{proof}

\subsection{Slater's condition for $\mathcal{P}_k$ and $\mathcal{V}_1$}\label{sec:interior}

We now give more concrete and direct arguments for Slater's condition in the specific situations of $\mathcal{P}_k$ and $\mathcal{V}_1$ that we consider in this paper.

\begin{proposition}[Slater's condition for $\mathcal{P}_k$] \label{prop:slaters_condition_omega_k}
    Let $A$ be in the interior of $\hull(\mathcal{P}_k)$. Then there is a density function on $\mathcal{P}_k$ which is in the interior of the constraints of $\primal_k(A)$.
\end{proposition}
\begin{proof}
    Define $P := \diag(1,\ldots,1,0,\ldots,0) \in \mathcal{P}_k$ and $P^\perp := I - P$. $Y = \alpha P^\perp - \beta P$ for some $\alpha,\beta > 0$ to be determined. Hence, $Y$ is $U(k) \times U(n-k)$-invariant, and therefore by the unitary invariance of $\mu_k$ we have the following for any $U \in U(k) \times U(n-k)$:
    \[
        B := \int X e^{-\langle Y, X \rangle} d\mu_k(X) = \int UXU^* e^{-\langle Y, UXU^* \rangle} d\mu_k(X) = U\left(\int X e^{-\langle Y, X \rangle} d\mu_k(X)\right)U^* = UBU^*.
    \]
    That is, $B$ is $U(k) \times U(n-k)$-invariant, and therefore $B = \gamma P^\perp + \delta P$ for some $\gamma,\delta \in \R$.

    By picking $\alpha,\beta > 0$ large with $\beta \ll \alpha$, the mass of $e^{-\langle Y, X \rangle} d\mu_k(X)$ becomes concentrated at $X = P$. If we also normalize by multiplying the density by $e^{-c\langle I_n, X \rangle}$ for appropriate values of $c \in \R$, we in fact have that $B$ approaches $P$ as $\alpha,\beta \to \infty$. Combining this with the form that $B$ must take means that for every $\epsilon > 0$, there exist $\alpha_\epsilon$ and $\beta_\epsilon$ such that the corresponding measure is normalized and the corresponding value of $B$ is equal to $(1-\epsilon)P + \frac{k\epsilon}{n}I_n$. We refer to this matrix as $B_\epsilon$, and we refer to the corresponding $Y$ matrix as $Y_\epsilon$.

    Note also that for any unitary $U \in U(n)$, the same argument holds for 
    $$U B_\epsilon U^* = (1-\epsilon)UPU^* + \frac{k\epsilon}{n} I_n$$ and $U Y_\epsilon U^*$. This then proves the result for $A = U B_\epsilon U^*$ for any $\epsilon > 0$ and any $U \in U(n)$. For any fixed $\epsilon > 0$, we further have:
    \[
        \hull(\{U B_\epsilon U^* ~:~ U \in U(n)\}) = (1-\epsilon)\hull(\mathcal{P}_k) + \frac{k\epsilon}{n} I_n =: \mathcal{P}_{k,\epsilon}.
    \]
    Hence, for any $A \in \mathcal{P}_{k,\epsilon}$, we can choose $U_1,\ldots,U_m \in U(n)$ such that $A = \frac{1}{m}\sum_{i=1}^m U_i B_\epsilon U_i^*$. Therefore:
    \[
        \int X \left(\frac{1}{m} \sum_{i=1}^m e^{-\langle U_i Y U_i^*, X \rangle}\right) d\mu_k(X) = \frac{1}{m}\sum_{i=1}^m U_i B_\epsilon U_i^* = A.
    \]
    Since $\frac{1}{m} \sum_{i=1}^m e^{-\langle U_i Y U_i^*, X \rangle} d\mu_k(X)$ is a convex combination of measures in the interior of the constraints of $\primal_k(A)$, this proves the result for all $A \in \mathcal{P}_{k,\epsilon}$. Letting $\epsilon \to 0$ then proves the result in full generality.
\end{proof}

\begin{proposition}[Slater's condition for $\mathcal{V}_1$] \label{prop:slaters_condition_GW}
    Let $A$ be in the interior of the PSD cone, and let $\mu$ be the pushforward of the Lebesgue measure $m$ though $x \mapsto xx^\top$. Then there is a density function on the set of rank-one real symmetric PSD matrices which is in the interior of the constraints of $\primal_\mu(A)$.
\end{proposition}
\begin{proof}
    Let $\nu_0(x)dm(x)$ be a Gaussian probability measure on $\R^n$ with covariance matrix $A$. This precisely means:
    \[
        \int xx^\top \nu_0(x) dm(x) = A.
    \]
    Let $\nu(X) d\mu(X)$ be the pushforward of $\nu_0(x) dm(x)$ through the map $x \mapsto xx^\top$. Then:
    \[
        \int X \nu(X) d\mu(X) = A.
    \]
    This $\nu(X)$ is the desired density function.
\end{proof}

\section{The Schur-Horn theorem} \label{sec:schur_horn}

We last discuss an idea that will useful to us in a number of parts of this paper. Generally, the idea is that the unitary invariance of $\mu_k$ allows us to often restrict to looking at diagonal matrices when considering the dual objective. The main observation is a corollary of the famous Schur-Horn theorem \cite{schur1923,horn1954}.

\begin{proposition}[Schur-Horn] \label{prop:schur_horn}
    If $D$ is a real diagonal matrix and $U$ is unitary, then the diagonal of $UDU^*$ is majorized by the diagonal of $D$.
\end{proposition}

\begin{corollary} \label{cor:schur_horn_opt}
    Given two real diagonal $n \times n$ matrices $D,D'$, we have the following:
    \[
        \inf_{U \in U(n)} \langle UDU^*, D' \rangle = \inf_{\sigma \in S_n} \langle \sigma D \sigma^*, D' \rangle.
    \]
    Here, $U(n)$ is the unitary group and $S_n$ is the subgroup of permutation matrices.
\end{corollary}
\begin{proof}
    Let $\sigma_0$ be the permutation matrix which minimizes $\langle \sigma D \sigma^*, D' \rangle$ over all permutation matrices. By majorization, for any $U$ the diagonal of $UDU^*$ can be written as a convex combination of the permutations of the diagonal of $D$. By linearity of $\langle \cdot, D' \rangle$, the value of $\langle UDU^*, D' \rangle$ must then be at least the value of $\langle \sigma_0 D \sigma_0^*, D' \rangle$.
\end{proof}

\begin{corollary}
    Let $A$ be a diagonal trace-$k$ PD matrix. Then:
    \[
        \inf_{Y, \text{Hermitian}} F_A(Y) = \inf_{Y, \text{real diagonal}} F_A(Y).
    \]
\end{corollary}
\begin{proof}
    Recall:
    \[
        F_A(Y) = \langle Y, A \rangle + \log \int e^{-\langle Y, X \rangle} d\mu_k(X).
    \]
    To prove the result, we only need to show that for any fixed real diagonal matrix $D$ we have:
    \[
        \inf_{U \in U(n)} F_A(UDU^*) = \inf_{\sigma \in S_n} F_A(\sigma D \sigma^*).
    \]
    Since the integration part of $F_A(Y)$ is unitarily invariant, this is then equivalent to:
    \[
        \inf_{U \in U(n)} \langle UDU^*, A \rangle = \inf_{\sigma \in S_n} \langle \sigma D \sigma^*, A \rangle.
    \]
    Since $A$ is diagonal, this follows from the previous corollary.
\end{proof}

\section{Closeness of the approximate distribution} \label{sec:closeness}

Let $\mu_1$ and $\mu_2$ two probability measures on $\Omega$, given as density functions with respect to a base measure $\mu$. The \emph{KL divergence} between $\mu_1$ and $\mu_2$ is defined as
\[
    D_\mathrm{KL}(\mu_1\|\mu_2) := \int_\Omega \mu_1(X)\log\left(\frac{\mu_1(X)}{\mu_2(X)}\right) d\mu(X).
\]
With this we follow the proof of Lemma A.4 in \cite{SinghV14} to obtain the following.

\begin{lemma}
    Let $Y^\star$ be the optimal solution to the dual objective function
    \[
        F_A(Y) = \langle Y, A \rangle + \log \int_\Omega e^{-\langle Y, X \rangle d\mu(X)}.
    \]
    Further, let $Y^\circ$ be such that $F_A(Y^\circ) \leq F_A(Y^\star) + \epsilon$. If $\mu^\star$ and $\mu^\circ$ are the probability distributions associated to $Y^\star$ and $Y^\circ$ respectively, then
    \[
        D_\mathrm{KL}(\mu^\star \| \mu^\circ) = F_A(Y^\circ) - F_A(Y^\star) \leq \epsilon.
    \]
\end{lemma}
\begin{proof}
    By assumption we have $F_A(Y^\circ) - F_A(Y^\star) \leq \epsilon$, which implies
    \[
        \langle Y^\circ - Y^\star, A \rangle + \log \int_\Omega e^{-\langle Y^\circ, X \rangle} d\mu(X) - \log \int_\Omega e^{-\langle Y^\star, X \rangle} d\mu(X) \leq \epsilon.
    \]
    The density functions of the distributions associated to $Y^\star$ and $Y^\circ$ can be given as
    \[
        \mu^\circ(X) := \frac{e^{-\langle Y^\circ, X \rangle}}{\int_\Omega e^{-\langle Y^\circ, X \rangle} d\mu(X)} \qquad \text{and} \qquad \mu^\star(X) := \frac{e^{-\langle Y^\star, X \rangle}}{\int_\Omega e^{-\langle Y^\star, X \rangle} d\mu(X)}.
    \]
    Since $Y^\star$ is the optimal solution (and hence $A$ is proportional to $\int_\Omega X e^{-\langle Y^\star, X \rangle} d\mu(X)$), we can compute the KL divergence as
    \[
    \begin{split}
        D_\mathrm{KL}(\mu^\star \| \mu^\circ) &= \frac{\int_\Omega e^{-\langle Y^\star, X \rangle} \left[\log\left(\frac{\int_\Omega e^{-\langle Y^\circ, Z \rangle} d\mu(Z)}{\int_\Omega e^{-\langle Y^\star, Z \rangle} d\mu(Z)}\right) + \langle Y^\circ - Y^\star, X \rangle\right] d\mu(X)}{\int_\Omega e^{-\langle Y^\star, X \rangle} d\mu(X)} \\
            &= \log\left(\frac{\int_\Omega e^{-\langle Y^\circ, X \rangle} d\mu(X)}{\int_\Omega e^{-\langle Y^\star, X \rangle} d\mu(X)}\right) + \left\langle Y^\circ - Y^\star, \frac{\int_\Omega X e^{-\langle Y^\star, X \rangle} d\mu(X)}{\int_\Omega e^{-\langle Y^\star, X \rangle} d\mu(X)} \right\rangle \\
            &= \left[\langle Y^\circ, A \rangle + \log \int_\Omega e^{-\langle Y^\circ, X \rangle} d\mu(X)\right] - \left[\langle Y^\star, A \rangle + \log \int_\Omega e^{-\langle Y^\star, X \rangle} d\mu(X)\right] \\
            &= F_A(Y^\circ) - F_A(Y^\star) \leq \epsilon.
    \end{split}
    \]
\end{proof}

\noindent
As in Corollary A.5 of \cite{SinghV14}, we use the previous result to obtain bounds on the approximate optimal distribution and on the marginals of this distribution.

\begin{corollary}
    Let $Y^\star$ be the optimal to the dual objective function $F_A(Y)$ with domian $\Omega$ and measure $\mu$ as in the previous lemma, and let $Y^\circ$ be such that $F_A(Y^\circ) \leq F_A(Y^\star) + \epsilon$. If $\mu^\star$ and $\mu^\circ$ are the probability distributions associated to $Y^\star$ and $Y^\circ$ respectively, then
    \[
        \|\mu^\star - \mu^\circ\|_\mathrm{TV} \leq \sqrt{2\epsilon}.
    \]
\end{corollary}
\begin{proof}
    The result follows from the previous lemma and the following well-known inequality (see e.g. \cite{cover2012elements}, Lemma 12.6.1, pp. 300-301) relating KL divergence and total variation distance:
    \[
        \|\mu^\star - \mu^\circ\|_\mathrm{TV} \leq \sqrt{2 \cdot D_\mathrm{KL}(\mu^\star\|\mu^\circ)}.
    \]
\end{proof}

\end{document}